\documentclass[11pt]{report}
\newcommand{\hideshow}[1]{{\mbox{}}}
\usepackage{multicol}
\usepackage{amsmath}
\usepackage{proof}
\usepackage{stmaryrd}
\usepackage{latexsym}
\usepackage{amssymb}
\usepackage{amsfonts}
\usepackage{esvect}
\usepackage{alltt}
\usepackage{amsthm}
\usepackage{bussproofs}
\usepackage{proof}
\usepackage{authblk}
\usepackage[all]{xy}
\usepackage{prettyref}
\usepackage{tabularx}
\usepackage{listings}
\usepackage{color}
\usepackage{wrapfig}
\usepackage{graphicx}
\usepackage{authblk}
\usepackage{url}
\usepackage[title]{appendix}

\definecolor{mygreen}{rgb}{0,0.6,0}
\definecolor{mygray}{rgb}{0.5,0.5,0.5}
\definecolor{mymauve}{rgb}{0.58,0,0.82}

\newcommand{\p}[1]{\mbox{$[\![#1]\!]$}}

 \newcommand{\reguno}[2]
  {
  $  \ \ \frac{\textstyle #1}{\textstyle #2} $
  }

\newcommand{\E}[2]{\ensuremath{{\epsilon}}}

\newcommand {\eg}        {{\textit{e}.\textit{g}.}}

\newcommand {\ie}        {{\textit{i}.\textit{e}.}}

\newtheorem{definition}{Definition}
\newtheorem{theorem}{Theorem}

\newtheorem{proposition}{\bf Proposition}

\newtheorem{lemma}{Lemma}

\newcommand{\bang}{\mbox{!}}

\begin{document}
\lstset{language=Erlang}
\title{$\lambda$-calculus and Reversible Automatic Combinators}

\author[1]{Alberto Ciaffaglione}
\author[1]{Furio Honsell}
\author[1]{Marina Lenisa}
\author[1]{Ivan Scagnetto}
\affil[1]{Department of Mathematics, Computer Science, and Physics\\
  University of Udine, Italy\\
  {\scriptsize\texttt{\{alberto.ciaffaglione,furio.honsell,marina.lenisa,ivan.scagnetto\}@uniud.it}}}

\maketitle

\begin{abstract}
 In 2005, S. Abramsky introduced various {\em linear/affine combinatory algebras} consisting of {\em partial involutions} over a suitable formal language, in order to discuss {\em reversible computation} in a game-theoretic setting. These algebras  arise as instances of the general paradigm explored by E.~Haghverdi, called ``Abramsky's Programme'', which amounts to defining a {\em $\lambda$-algebra} starting from a so called {\em GoI Situation} in a ``traced symmetric monoidal category''. We recall that GoI is the acronym for {\em ``Geometry of  Interaction''}. This was invented by J.Y.Girard in order to model, in a language-independent way, the fine semantics of {\em Linear Logic}.

In this paper, we investigate Abramsky's construction from the point of view of the model theory of $\lambda$-calculus. We focus on the \emph{strictly linear} and the {\em strictly affine} parts of Abramsky's Affine Combinatory Algebras, and we outline briefly, at the end, how  the {\em full} algebra can be encompassed.

The gist of our approach is that the interpretation of a term based on involutions is ``dual'' to the {\em principal type} of the term, with respect to the {\em simple types discipline}  for a linear/affine  {\em $\lambda$-calculus}. Thus our analysis unveils three conceptually independent, but ultimately equivalent, accounts  of {\em application} in the $\lambda$-calculus: {\em $\beta$-reduction}, the GoI application of involutions based on symmetric feedback (Girard's {\em Execution Formula}), and {\em unification} of principal types. Somehow surprisingly, this equivalence had not been hitherto pointed out.

Our result permits us to provide an answer, in the strictly affine case,  to the question raised in \cite{Abr05} of characterising the partial involutions arising from bi-orthogonal pattern matching automata,  which are denotations of affine combinators, and it points to the answer to the full question.
Furthermore, we prove that the strictly linear combinatory algebra of partial involutions is a strictly linear $\lambda$-algebra, albeit not a combinatory model, while both the strictly affine combinatory algebra and the full affine combinatory algebra are not.\\
In order to check all the necessary equations involved in the definition of affine $\lambda$-algebra, we implement in  Erlang  application of involutions,  as well compilation of $\lambda$-terms as combinators and their interpretation as involutions.
 \end{abstract}

\section{Introduction}
In \cite{Abr05}, S. Abramsky discusses {\em Reversible Computation} in a game-theoretic setting. In par\-ti\-cu\-lar, he introduces various kinds of reversible {\em  pattern-matching automata} whose behaviour can be described in a {\em finitary way} as partial injective functions,  actually {\em involutions}, over a suitable language.  These  yield affine {\em combinatory algebras} w.r.t. a notion of application between automata, for which the flow of control is analogous to the one between history-free strategies in game models. Similar constructions appear in various papers by S. Abramsky, \eg\ \cite{AHS02,AL05}, and are special cases of a general categorical paradigm explored by E. Haghverdi \cite{Hagh00}, Sections~5.3,~6, called ``Abramsky's Programme''. This Programme amounts to defining a {\em $\lambda$-algebra} starting from a {\em GoI Situation} in a ``traced symmetric monoidal category'', where application arises  from {\em symmetric feedback}/Girard's {\em Execution Formula}. We recall that GoI is the acronym for ``Geometry of Interaction'', an approach invented by J.~Y.~Girard~\cite{Girard89,Girard90} in order to model, in a language-independent way, the fine semantics of {\em Linear Logic}.

In this paper, we discuss Abramsky's algebras from the point of view of the model theory of $\lambda$-calculus.
We focus on the \emph{strictly linear} and \emph{strictly affine} parts of Abramsky's affine algebras. In particular,
we consider strictly linear and strictly affine  combinatory logic, their $\lambda$-calculus counterparts, and their models, \ie\ {\bf BCI}{\em -combinatory algebras} and {\bf BCK}{\em -combinatory algebras}.
For each calculus we discuss also the corresponding notion of {\em $\lambda$-algebra}. This notion was originally introduced by D. Scott for the standard  $\lambda$-calculus as the appropriate notion of categorical model for the calculus, see Section~5.2 of~\cite{Bar84}.

The gist of our approach is that the {\em Geometry of Interaction} interpretation of a $\lambda$-term in Abramsky's model of partial involutions  is ``dual'' to the {\em principal type} of that term w.r.t. the \emph{simple  types discipline}  for a linear/affine {\em $\lambda$-calculus}.

In particular, we define an algorithm which, given a principal type of a $\lambda$-term, reads off the partial involution corresponding to the interpretation of that term. Conversely, we show how to extract a principal type from a partial involution, possibly not corresponding to any $\lambda$-term.
Moreover, we show that the principal type of an affine $\lambda$-term provides a {\em dual} characterisation of the partial involution interpreting the term in Abramsky's model.
The overall effect of the GoI notion of application amounts to {\em unifying} the left-hand side of the principal type  of the operator with the principal type  of the operand, and applying the resulting substitution to the right hand side of the operator. Hence, the notion of application between partial involutions, corresponding to $\lambda$-terms $M$ and $N$,  can be explained as computing the involution corresponding to the principal type of $MN$, given the principal types of $M$ and $N$.

This analysis, therefore, unveils three conceptually independent, but ultimately equivalent, accounts  of {\em application} in the $\lambda$-calculus: {\em $\beta$-reduction}, the GoI application of involutions based on symmetric feedback/Girard's {\em Execution Formula}, and {\em unification} of principal types. Somehow surprisingly, this equivalence had not been hitherto pointed out.

Our results provide an answer, for the strictly affine part, to the question raised in \cite{Abr05} of characterising the partial involutions
(arising from bi-orthogonal pattern matching automata),
which are denotations of  combinators. We show that these are precisely those partial involutions whose corresponding principal type is the principal type of a $\lambda$-term.
In our view, this insight sheds new light on the deep nature of Game Semantics itself.

We prove, furthermore, that the strictly linear combinatory algebra of partial involutions is also a strictly linear $\lambda$-algebra, albeit not a combinatory model, while both the strictly affine combinatory algebra and the affine combinatory algebra are not $\lambda$-algebras. We also show  that the last step of Abramsky's programme, namely the one taking from a linear/affine combinatory algebra to a {\em $\lambda$-algebra}, is not immediate, since in general combinatory algebras cannot be quotiented non trivially to obtain $\lambda$-algebras.

In order to check all the necessary equations of $\lambda$-algebras, we implement in Erlang~\cite{Erlang1,Erlang2}  application of involutions, as well compilation of $\lambda$-terms as combinators and their interpreation as involutions.

In the final remarks,  we briefly outline how to extend the above results to the full affine $\lambda$-calculus (\ie\ the $\lambda$-calculus extended with a $!$-operator and a corresponding pattern-abstraction) w.r.t.\! an extension of the \emph{intersection type discipline} with a $!_u$-constructor. Intersection types originated in \cite{BCD} and have been utilised in discussing games in a different approach also in \cite{DGHL,DGL13}.

A clarification is in order as far as the use of the term {\em reversible} by Abramsky and, hence, in the title of the present, although we do not develop it further here. Abramsky's algebras provide indeed {\em universal models of reversible computation}, but not in the simplistic sense that application between combinators is itself reversible. What is reversible is the evaluation of the  partial involution interpreting a combinator. Since the partial involutions interpreting, say 0 and 1, have different behaviours on a simple ``tell-tale'' word, we can {\em test reversibly} any characteristic function expressed in terms of combinators, without evaluating the overall combinator. The {\em finitary descriptions} of the partial involutions interpreting the combinators, therefore, have full right to be called {\em reversible combinators}.

\smallskip

\noindent {\bf Synopsis.} In Section~\ref{sacl}, we introduce the strictly linear and strictly affine versions of: combinatory logic, $\lambda$-calculus, combinatory algebra, and combinatory model, and we isolate the equations for the strictly linear and strictly affine combinatory algebras to be $\lambda$-algebras. In Section~\ref{tre}, we provide a type discipline for the strictly linear and strictly affine $\lambda$-calculus, and we define a corresponding notion of principal type. In Section~\ref{qua}, we recall Abramsky's combinatory algebra of partial involutions, and we provide a characterisation of partial involutions via principal types, in the strictly affine case. Furthermore, we prove that partial involutions are a strictly linear $\lambda$-algebra but they are not a strictly affine (nor an affine) $\lambda$-algebra. In Section~\ref{erlang}, we discuss  the implementation in Erlang of the application between partial involutions, and compilation and interpretation of $\lambda$-terms. Concluding remarks appear in Section~\ref{fin}.
The Appendix includes the detailed Erlang programs implementing compilations and effective operations on partial involutions.

\section{Strictly Linear Notions and their Strictly Affine Extensions }\label{sacl}

First we introduce {\em strictly linear} versions of combinatory logic, $\lambda$-calculus,  and combinatory algebras.
For each notion we also introduce  corresponding  {\em  strictly affine} extensions, which include  constant operations and the combinator {\bf K}.
These notions are the restrictions of the corresponding notions of  combinatory logic and $\lambda$-calculus \cite{Bar84} to the purely linear (affine) terms.

We assume the reader familiar with the basic notations and results in combinatory logic and $\lambda$-calculus, as presented \eg\ in \cite{Bar84}, and in \cite{Abr05}, but we try to be self-contained as much as possible.

\begin{definition}[Strictly Linear (Affine) Combinatory Logic]\label{sacl}
The language of {\em strictly linear (affine) combinatory logic} $\mathbf{CL^L}$  ($\mathbf{CL^A}$) is generated by variables $x, y, \ldots $ and constants, which include the distinguished constants (combinators) $B, C,  I$ (and $K$ in the affine case) and
it is closed under application, \ie:\medskip
\\
\reguno{M \in \mathbf{CL^X} \ \ \  N \in \mathbf{CL^X}
}{M \cdot\ N \in \mathbf{CL^X} } \ \ \ \ \ for $X\in \{ L,A \}$
 \medskip
\\ Combinators satisfy the following
equations (we associate $\cdot$ to the left and omit it when clear from the context):  \smallskip
  \\
\begin{tabular}{llll}
  $BMNP = M(NP)$ \ \ \ \ &       $IM = M$ \ \ \ \     &      $CMNP = (MP)N$   \ \ \ \ &   $K MN = M   $

\end{tabular}\\
\noindent where  $M,N,P$ denote terms  of combinatory logic.
\end{definition}
\begin{definition}[Strictly Linear (Affine) Lambda Calculus]
The language $\mathbf{\Lambda^L}$ ($\mathbf{\Lambda^A}$) of the {\em strictly linear  (affine) $\lambda$-calculus, \ie\ $\lambda^L$-calculus ($\lambda^A$-calculus)} is inductively defined from  variables $x,y,z, \ldots \in \mathit{Var}$, constants $c, \ldots \in \mathit{Const}$, and it is closed under  the following  formation rules:\medskip
\\
$\mathbf{\Lambda^L}$: \reguno{M \in \mathbf{\Lambda^L}\ \ \  N \in \mathbf{\Lambda^L} }{MN \in \mathbf{\Lambda^L}}\ \
\reguno{M \in \mathbf{\Lambda^L} \ \ {\cal E} (x, M)}{\lambda x.M \in \mathbf{\Lambda^L}} \ \ \ \smallskip
\\ $\mathbf{\Lambda^A}$:
 \reguno{M \in \mathbf{\Lambda^A}\ \ \  N \in \mathbf{\Lambda^A} }{MN \in \mathbf{\Lambda^A}}\ \
\reguno{M \in \mathbf{\Lambda^A} \ \ {\cal O} (x, M)}{\lambda x.M \in \mathbf{\Lambda^A}}
\medskip
\\
where ${\cal E} (x, M)$ means that the variable $x$ appears free  in $M$ {\em exactly once}.
\\
where ${\cal O} (x, M)$ means that the variable $x$ appears free in $M$ {\em at most once}.
\smallskip

The rules of the $\lambda^L$-calculus ($\lambda^A$-calculus) are the restrictions of the standard $\beta$-rule and  $\xi$-rule to strictly linear (affine) abstractions, namely: \medskip

\noindent $(\beta_L) \   (\lambda x. M)N = M[N/x]$ \ \ \ \
\noindent $(\xi_L)$
 \reguno{M = N \ \ {\cal E}(x,M)\ \  {\cal E}(x,N) }{\lambda x. M = \lambda x.N}.
\smallskip

\noindent
$(\beta_A) \  (\lambda x. M)N = M[N/x]$ \ \ \ \
\noindent $(\xi_A)$ \
 \reguno{M = N \ \ {\cal O}(x,M)\ \   {\cal O}(x,N)}{\lambda x. M = \lambda x.N}.

\medskip
All the remaining rules are the standard rules which make $=$ a congruence.
\end{definition}

\begin{proposition}
Well-formedness in $\mathbf{\Lambda^{L}}$ ($\mathbf{\Lambda^{A}}$), \ie\ strictly linear (affine) $\lambda$-abstractions are preserved under $\lambda$-reduction. The corresponding reduction calculi are  Church-Rosser.
\end{proposition}
\begin{proof}Routine.
\end{proof}

In the sequel of this section, for conciseness,  we discuss only the $\lambda^A$-calculus, since the corresponding notions/results carry over straightforwardly to the strictly linear version by simple restriction.

We start by specialising to the affine case the results in \cite{Bar84} on the encoding of $\lambda$-calculus into combinatory logic.

\begin{definition}
We define two homomorphisms w.r.t. application:
\\ (i) $(\ )_{\lambda^A}: \mathbf{CL^A}\rightarrow \mathbf{\Lambda^A}$, given a term $M$ of   $ \mathbf{CL^A}$, yields the term of $ \mathbf{\Lambda^A}$ obtained from $M$
by replacing each combinator with the corresponding $ \mathbf{\Lambda^A}$-term as follows
\\ \begin{tabular}{llll}
 $ (B)_{\lambda^A}= \lambda x y z. x(yz) $ \ \ \   & $(I)_{\lambda^A}= \lambda x. x $\ \ \
 $(C)_{\lambda^A} = \lambda x y z.(xz)y   $\ \ \    & $   (K)_{\lambda^A} = \lambda x  y . x   $
\\
 \end{tabular}
 \\ (ii)  $(\ )_{CL^A}: \mathbf{\Lambda^A}\rightarrow \mathbf{CL^A}$, given a term  $M\in \Lambda^A$, replaces each $\lambda$-abstraction
 by a $\lambda^*$-abstraction.  Terms with $\lambda^*$-abstractions  amount to  $\mathbf{CL^A}$-terms via the \emph{Abstraction Operation} defined below. \end{definition}

 \begin{definition}[Strictly Affine Abstraction Operation] \label{op} The following operation,  defined by induction on  $M\in \mathbf{CL^A}$, provides an encoding of
 $\lambda^A$-calculus into $\mathbf{CL^A}$:
\\  $\lambda^{*}x. x  = I$ \ \ \  $\lambda^{*}x. c  = Kc$ \ \ \ $\lambda^{*}x. y  = Ky$\ , for $c\in \mathit{Const}$, $x\neq y$
\\ $\lambda^* x. MN = \begin{cases}
C (\lambda^* x.M) N & \mbox{ if } x\in FV(M),
\\ BM (\lambda^* x.N)  & \mbox{ if }  x\in FV(N),
\\ K(MN) & \mbox{ otherwise.}
\end{cases}$
\end{definition}

\begin{theorem}[Strictly Affine Abstraction Theorem] \label{absth}
For all terms $M,N \in \mathbf{CL^A}$, $(\lambda^* x.M) N = M[N/x]$.
\end{theorem}
\begin{proof}By strightforward induction on the definition of $\lambda^*$.
\end{proof}

The notion of {\em strictly linear (affine)  combinatory algebra}, or {\bf BCI}-algebra ({\bf BCK}-algebra) is the restriction of the notion of {\em combinatory algebra} to
strictly linear (affine) combinatory logic:

\begin{definition}[Strictly Linear (Affine) Combinatory Algebra, {\bf BCI}-alge\-bra ({\bf BCK}- algebra)]\label{lca}
\hfill
\\ (i) A {\em strictly linear (affine) combinatory algebra, SLCA, (SACA)}  $\mathcal{ A}= (A, \cdot)$ is
an
applicative structure
$(A,
\cdot)$, and distinguished
elements
(combinators) $B, C,  I$ (and $K$ in the affine case) sati\-sfying the following
equations: for all $x,y,z \in A$,
\smallskip

\begin{tabular}{llll}
  $ Bxyz =_{\mathcal A} x(yz) $ \ \ \ \ \ \ & $  Ix =_{\mathcal A} x  $    \ \ \ \ \ \ &
$  Cxyz =_{\mathcal A} (xz)y   $ \ \ \ \   \ \ &  $   K x  y =_{\mathcal A} x  $
\end{tabular}
\\ (ii) For a  strictly linear (affine) combinatory algebra $\mathcal{ A}$, we define  $\p{\ }_{\mathcal{A}}: \mathbf{CL^A}\rightarrow \mathcal{A}$ as  the natural interpretation of closed terms of  $ \mathbf{CL^L}$  ($ \mathbf{CL^A}$) into $\mathcal{A}$.
\\ (iii) For a  strictly linear (affine) combinatory algebra $\mathcal{ A}$, we define the set of strictly linear (affine) combinatory terms $\mathcal{T}(\mathcal{ A})$ as the extension of $\mathbf{CL^L}$ ($\mathbf{CL^A}$)
with constants $c_a$ for $a \in { A}$.
\end{definition}

In what follows, when clear from the context, we will simply write $=$ in place of $=_{\mathcal A}$.

As we did earlier for the syntactic notions, we will discuss semantic notions only for the strictly affine case. If not stated explicitly, the corresponding notions/theorems  carry over straightforwardly, {\em mutatis mutandis}, to the strictly linear case.

First we introduce {\em strictly affine $\lambda$-algebras}. These were
originally introduced by D. Scott for  standard  $\lambda$-calculus as the appropriate notion of categorical model for the calculus, see Definition~5.2.2(i) of~\cite{Bar84}.

\begin{definition}[Strictly Affine $\lambda$-algebra]
A SACA  $\mathcal{ A}$  is
 a \emph{strictly affine $\lambda$-algebra} if, for all closed $M,N \in \mathcal{T}(\mathcal{ A})$,
\[ \vdash (M)_{\lambda^A} =_{\lambda^A} (N)_{\lambda^A}  \ \Longrightarrow\ \p{M}_{\mathcal{ A}}=\p{N}_{\mathcal{ A}}  \ , \]
 where $ =_{\lambda^A}$ denotes provable equivalence on $\lambda$-terms, and $\p{\ }_\mathcal{ A}$ denotes (by abuse of notation) the natural extension to terms in $\mathcal{T}(\mathcal{ A})$ of the interpretation
$\p{\ }_{\mathcal{ A}}: \mathbf{CL^A}\rightarrow \mathcal{A}$.
 \end{definition}

Given a {\bf BCIK}-algebra, there exists a smallest quotient giving rise to a (possibly trivial) strictly affine $\lambda$-algebra, namely:

\begin{definition}\label{sette}
Let ${\mathcal A}= (A, \cdot)$ be a SACA. For all $a,b\in A$, we define $a\equiv_{\mathcal A} b$ if and only if there exist {\em closed} $M, N \in   {\mathcal T}({\mathcal A})$ such that $a= \p{M}_{\mathcal A}$, $b=  \p{N}_{\mathcal A}$, and
 $(M)_{\lambda^A} =_{\lambda^A} (N)_{\lambda^A}$.
\end{definition}

We have:

\begin{proposition}\hfill
\\ (i) Not all SACA's are strictly affine $\lambda$-algebras.
\\ (ii) Let ${\mathcal A}= (A, \cdot)$ be a SACA. Then the quotient $ (A/\equiv_{\mathcal A}, \cdot_{\equiv_{\mathcal A}})$ is a strictly affine $\lambda$-algebra.
\\ (iii) Not all non-trivial SACA's can be quotiented to a non-trivial strictly affine $\lambda$-algebra.
\end{proposition}
\begin{proof}
(i) A trivial example is the closed term model of strictly affine combinatory logic, \ie\ the quotient of closed terms under equality, \eg\ CKK $\neq$ I. A more subtle example is the algebra of partial involutions ${\cal P}$ discussed in Section \ref{qua}.
\\ (ii) $\equiv_{\mathcal A}$ is a congruence w.r.t. application, since $=_{\lambda}$ is a congruence. Then the thesis follows from definitions.
\\ (iii) Consider the closed term model of standard combinatory algebra induced by the equations $(SII)(SII) = I $ and
 $(S(BII)(BII))(S(BII)(BII))  = K $.
This is clearly a SACA.  $S$ is the standard combinator from Combinatory Logic, see \eg\ \cite{Bar84}. The lhs's are  terms reducing to themselves, and thus can be consistently (\emph{i.e.} without producing a trivial model) indipendently equated to whatever;  but they are equated to each other in any strictly affine $\lambda$-algebra. Hence any quotient of this term model to a strictly affine $\lambda$-algebra is trivial, because $I=K$. In the strictly linear case the argument has to be modified by taking the second equation to be \eg\ $(S(BII)(BII))(S(BII)(BII))  = B $.
\end{proof}

We give now the notion of {\em strictly affine combinatory model}. The corresponding one for standard $\lambda$-calculus was introduced by A. Meyer in his seminal paper \cite{Meyer}.

\begin{definition}[Strictly Affine Combinatory $\lambda$-model]
\label{due}
 A  SACA $\mathcal{ A} $ is a \emph{strictly affine combinatory $\lambda$-model}
 if there exists a \emph{selector} combinator $\epsilon$ such
 that, for all $x,y\in A$, $\epsilon xy = xy$ and $(\forall z.\ xz = yz)\Longrightarrow \epsilon x = \epsilon y$.
 \end{definition}

\begin{proposition}
 Not all  strictly affine $\lambda$-algebras are strictly  affine combinatory $\lambda$-models.
\end{proposition}
\begin{proof} In the case of the standard combinatory logic, and hence strictly affine combinatory logic,  this is implied by the well known conjecture of Barendregt on the failure of the $\omega$-rule, finally disproved by G. Plotkin using {\em universal generators}, (see \cite{Bar84}, Section 17.3-4). Theorem~\ref{Nocomb} below provides such a counterexample for the strictly linear case, namely the algebra of partial involutions ${\cal P}$.
\end{proof}

Curry was the first to discover that $\lambda$-algebras have {\em purely equational definitions}. We give corresponding results for strictly linear  and affine combinatory logic, which, although natural, are probably original. The significance of the following theorem is that a finite number of equations involving combinators, $A^\beta $, are enough to ensure that the congruence on $\mathbf{CL^A}$-terms is closed under the $\xi_A$-rule, as a rule of proof,  namely if $\mathbf{CL^A}+ A^\beta \vdash M=N$ then $\vdash \p{\lambda^* x.M}_{\cal A} = \p{\lambda^* x.N}_{\cal A}$.

\begin{theorem}\label{th:eqforSALA}
A SACA  $\cal A$ satisfying the following  sets of equations is a strictly  affine $\lambda$-algebra:
\begin{itemize}
\item  \begin{tabular}{ll}
 $ B= \lambda^* x y z. x(yz) = \lambda^* x y z. Bxyz $ &
\\
 $  C = \lambda^* x y z.(xz)y = \lambda^* x y z.C x yz    $\ \ \ \   \ \ \ \ &
\\
$  I = \lambda^* x. x = \lambda^* x. I x $ \ \ \ \ \ \ \ \ &
\\
$   K = \lambda^* x  y .x = \lambda^* xy. K x  y $
 \end{tabular}

\item equations for $\lambda^* x. IP = \lambda^* x.P$ to hold:  $\lambda^* y. BIy = \lambda^* yz.yz$

\item equations for $\lambda^* x. BPQR=\lambda^* x. P(QR)$ to hold:
\begin{itemize}
\item $\lambda^* u v w. C(C(BB u) v) w = \lambda^* u v w. C u (v w)$

\item  $\lambda^* u v w. C(B(B u) v) w = \lambda^* u v w. B u (C v w)$

\item $\lambda^* u v w. B(B u v) w = \lambda^* u v w. B u (B v w)$
\end{itemize}
\item for $\lambda^* x. CPQR=\lambda^* x. PRQ$ to hold:
\begin{itemize}
\item $\lambda^* u v w. C(C(BC u) v) w =
\lambda^* u v w. C (C u w) v $
\item  $\lambda^* u v w. C(B(C u) v) w =
\lambda^* u v w. B (u w) v$
\item  $\lambda^* u v w. B(C u v) w = \lambda^* u v w. C (B u w) v$
\end{itemize}
\item for $\lambda^* x.KPQ = \lambda^* x.P$ to hold :
\begin{itemize}
\item $\lambda^* xy. C(BKx)y = \lambda^* xyz. xz$
\item $\lambda^* xy. B(Kx)y = \lambda^* xyz. x$
\end{itemize}
\item 2 more equations are necessary for $K$ in dealing with $\xi$ over axioms:
\begin{itemize}
\item $\lambda^* xy. Bx(Ky) = \lambda^* xy. K(xy) $
\item $\lambda^* xy. C(Kx)y = \lambda^* xy. K(xy) $
\end{itemize}
\end{itemize}
\end{theorem}
\begin{proof}(Sketch)
\noindent The proof follows closely the argument in \cite{Bar84}, Section 7.3. The equations allow for proving that $\mathbf{CL^A}$ is closed under the $\xi_A$-rule. For each combinator we have therefore as many equations as possible branches in the Abstraction Operation. At the very end, suitable $\lambda^*$-abstractions need to be carried out in order to remove the parameters.
\end{proof}

The corresponding theorem in the {\em strictly linear}
case is obtained by deleting all the equations referring to K.

\section{Strictly Linear  and Affine Type Discipline for the $\lambda$-calculus}\label{tre}

In this section we introduce the key type-theoretic tools for understanding the fine structure of partial involutions, namely {\em principal simple type schemes}. Principal types were introduced by Hindley see \eg\ \cite{Hindley69} but with a different purpose. We discuss the strictly linear and strictly affine cases separately, because they exhibit significantly different properties.

\begin{definition}[Simple Types]
$(\mathit{Type}\ni)\ \mu::= \alpha\ |\ \mu \rightarrow \mu\ $, where $\alpha\in \mathit{TVar}$  denotes a type variable.
\end{definition}

\begin{definition}[Strictly Linear   Type Discipline]\label{discipline} The \emph{strictly linear  type system} for the  $\lambda^L$-calculus is given by the following set of rules for assigning {\em simple  types} to terms of ${\mathbf \Lambda^L}$.
 Let $\Gamma, \Delta$ denote environments, \ie\ {\em sets} of the form $\Gamma = x_1: \mu_1, \ldots , x_m:\mu_m$, where each variable in  $\mathit{dom}(\Gamma) = \{ x_1, \ldots , x_m\}$ occurs exactly once:

\medskip

\reguno{}{x : \mu\vdash_L  x:\mu}  \ \ \
\reguno{x \in FV(M)\ \ \ \ \  \Gamma,  x:\mu \vdash_L M:\nu}{\Gamma\vdash_L \lambda x. M : \mu\rightarrow \nu}  \medskip

\reguno{\Gamma\vdash_L M: \mu \rightarrow \nu\ \ \ \ \Delta\vdash_L N:\mu\ \ \  (dom(\Gamma)\cap \mathit{dom}(\Delta))=\emptyset}{\Gamma, \Delta \vdash_L MN: \nu}.
\end{definition}

We introduce now the crucial notion of {\em principal type scheme}:

\begin{definition}[Principal Type Scheme]\label{principal}
Given a $\lambda^L$-term $M $, the judgement $\Gamma \Vdash_L M : \sigma$ denotes that $\sigma$ is the {\em principal type scheme} of $M$:\medskip

\begin{small}
\noindent
\reguno{}{x:\alpha\Vdash_L x:\alpha} \ \ \ \
\reguno{x \in FV(M)\ \ \ \ \  \Gamma, x:\mu\Vdash_L M:\nu}{\Gamma\Vdash_L \lambda x. M : \mu\rightarrow \nu}
\medskip

\noindent \hspace*{-0.7cm} \reguno{\begin{tabular}{c} $\Gamma\Vdash_L M: \mu \ \ \ \Delta\Vdash_L N:\tau \ \ \ (\mathit{dom}(\Gamma) \cap \mathit{dom}(\Delta))= \emptyset  \ \ \  (\mathit{TVar}( \Gamma)\cap \mathit{TVar}( \Delta))=\emptyset   $
\\    $    (\mathit{TVar}( \mu)\cap \mathit{TVar}( \tau))=\emptyset  \ \ \   U' = MGU(\mu, \alpha \rightarrow \beta) \ \ \  U= MGU(U'(\alpha), \tau)\ \ \  \alpha, \beta \ \mathit{fresh}$  \end{tabular} }{U(\Gamma, \Delta) \Vdash_L MN: U\circ U' (\beta)}
 \medskip\\
\end{small}

\noindent where $MGU$ gives the most general unifier, and it is defined (in a standard way) below. By abuse of notation, $U$ denotes also the substitution on contexts induced by $U$.
\end{definition}

\begin{definition}[$MGU(\sigma,\tau)$]
Given two types $\sigma$ and $\tau$, the partial algorithm $MGU$ yields a substitution  $U$ on type variables (the identity almost everywhere) such that $U(\sigma)=U(\tau)$: \medskip

\noindent \reguno{MGU(\alpha,\tau)=U \ \ \alpha\in TVar \ \ \tau\not\in TVar}{MGU(\tau,\alpha) = U}
\ \ \
\reguno{\alpha\in TVar\ \ \  \alpha\not\in \tau}{MGU(\alpha,\tau) = id[\tau/ \alpha]}
\medskip

\noindent \reguno{MGU(\sigma_1,\tau_1)=U_1\ \ \ MGU(U_1 (\sigma_2), U_1(\tau_2))=U_2}{MGU(\sigma_1\rightarrow \sigma_2, \tau_1 \rightarrow \tau_2)= U_2 \circ U_1}\\

\noindent where $ U_2 \circ U_1$ denotes composition between the extensions of the substitutions to the whole set of terms; $id$ denotes the identical substitution.
\end{definition}
As is well known, the above algorithm yields a substitution which factors any other unifier, see \eg\ \cite{Rob65}.

The following theorem, which can be proved by induction on derivations, connects the systems defined above.

\begin{theorem}\label{theorem_principal}
For all $M \in \Lambda^L$:
\\ (i) if $\Gamma \Vdash_L M:\sigma$ and  $\Gamma' \Vdash_L M:\sigma'$ are derivable, then $\sigma =_{\alpha} \sigma'$ and $\Gamma =_{\alpha} \Gamma'$, \emph{i.e.}
$\mathit{dom}(\Gamma) = \mathit{dom}(\Gamma')$ and $x:\mu \in \Gamma, x:\mu' \in \Gamma' \ \Rightarrow\ \ \mu=_{\alpha} \mu'$.
\\ (ii) if $\Gamma \Vdash_L M:\sigma$ is derivable, then each type variable occurs at most twice in $\Gamma \Vdash_L M:\sigma$.
\\ (iii) if $\Gamma \Vdash_L M:\sigma$, then, for all substitutions $U$ such that each type variable occurs at most twice in $U(\Gamma), U(\sigma)$,  $U(\Gamma) \vdash_L M:U(\sigma)$;
\\ (iv) if $\Gamma \vdash_L M:\sigma$, then there exists a derivation $\Gamma' \Vdash_L M: \sigma'$ and a type substitution $U$, such that $U(\Gamma')=\Gamma$ and $U(\sigma')=\sigma$.
\end{theorem}

Here are some well known examples of principal types:

\begin{tabular}{lll}
$I$   &  $\lambda x.x$     &   $\alpha \rightarrow \alpha$\\
$B$   &  $\lambda xyz.x(yz)$     &   $(\alpha \rightarrow \gamma) \rightarrow (\beta  \rightarrow \alpha ) \rightarrow \beta \rightarrow \gamma$\\
$C$   &  $\lambda xyz.xzy$     &   $(\alpha \rightarrow \beta \rightarrow \gamma)  \rightarrow \beta \rightarrow \alpha \rightarrow \gamma$\\
\end{tabular}

\begin{theorem}[Strictly Linear Subject Conversion]\label{subjectreduction}
Let $M \in \Lambda^L$, $M =_{\beta^L} M'$, and  $\Gamma \Vdash_L M:\sigma$, then  $\Gamma \Vdash_L M':\sigma$.
\end{theorem}
\begin{proof}
First one proves
 subject conversion for $\vdash_L$, \ie:\\ $\Gamma \vdash_L M:\sigma \ \wedge\  M =_{\beta^L} M' \ \Longrightarrow\ \Gamma \vdash_L M':\sigma$.
This latter fact follows from:
\\ $\Gamma, x:\mu \vdash_L M:\nu \ \wedge\ \Delta \vdash_L N: \mu \ \Longleftrightarrow \ \Gamma, \Delta \vdash_L M[N/x]: \nu \ \wedge\ \Delta \vdash_L N: \mu$\ ,
\\ which can be easily proved by induction on $M$.
\\ Now, let $M =_{\beta^L} M'$ and $\Gamma \Vdash_L M:\sigma$. Then, by Theorem~\ref{theorem_principal}(iii), $\Gamma \vdash_L M:\sigma$, and by subject conversion
of $\vdash_L$, $\Gamma \vdash_L M':\sigma$. Hence, by Theorem~\ref{theorem_principal}(iv), there exist $U, \Gamma', \sigma'$ such that $\Gamma' \Vdash_L M':\sigma'$ and
 $U(\Gamma')= \Gamma$, $U(\sigma')= \sigma$. But then, by Theorem~\ref{theorem_principal}(iii),  $\Gamma' \vdash_L M':\sigma'$, and by subject conversion of $\vdash_L$,
 $\Gamma' \vdash_L M:\sigma'$. Hence, by Theorem~\ref{theorem_principal}(iv), there exist $U', \Gamma'', \sigma''$ such that  $\Gamma'' \Vdash_L M:\sigma''$ and
 $U'(\Gamma'')=\Gamma'$, $U'(\sigma'') = \sigma'$. Finally, by Theorem~\ref{theorem_principal}(i), $\Gamma =_{\alpha} \Gamma''$ and $\sigma =_{\alpha} \sigma''$, and hence also
 $\Gamma =_{\alpha} \Gamma'$ and $\sigma =_{\alpha} \sigma'$.
\end{proof}

\subsection{The Strictly Affine Case: Discussion}\label{noaffine}

The extension to the $\lambda^A$-calculus of Definition \ref{discipline} is  apparently unproblematic. We can just add  to $\vdash$ the natural rule
\noindent \reguno{\Gamma  \vdash_A M : \nu \ \ x \not\in \mathit{dom}(\Gamma) \ \ \  TVar( \mu) \ \mbox{fresh}}{\Gamma \vdash_A \lambda x.M : \mu \rightarrow \nu}\smallskip\\
and its counterpart to $\Vdash$, namely \ \ \
\noindent \reguno{\Gamma  \Vdash_A M : \nu \ \  x \not\in \mathit{dom}(\Gamma) \ \ \  \alpha \ \mbox{fresh}}{\Gamma \Vdash_A \lambda x.M : \alpha \rightarrow \nu}.\smallskip \\
However, while we get type assignment systems which satisfy the extension  of Theorem~\ref{theorem_principal} to the affine case, the affine version of Theorem \ref{subjectreduction}, \ie\ subject conversion, fails. Namely, we cannot derive $  \Vdash_A \lambda x y z. (\lambda w. x)(yz)  : \alpha_1 \rightarrow \alpha_2 \rightarrow \alpha_3 \rightarrow \alpha_1$, but only $ \Vdash_A \lambda x y z. (\lambda w. x)(yz): \alpha_1 \rightarrow (\alpha_2 \rightarrow \alpha_3) \rightarrow \alpha_2 \rightarrow \alpha_1$, which is an instance of the former.
But we have $\vdash \lambda x y z.x: \alpha_1 \rightarrow \alpha_2 \rightarrow \alpha_3 \rightarrow \alpha_1$.
This is the key reason for the failure of the Affine Combinatory Algebra $\cal P$ defined  in Section \ref{qua} to be a strictly affine  $\lambda$-algebra.

\section{Abramsky's Model of Reversible Computation} \label{qua}

 S. Abramsky, in \cite{Abr05}, impressively exploits  the connection between {\em automata} and strategies and introduces various  reversible universal models of computation. Building on earlier work, \eg\ \cite{AL05,Hagh00}, S. Abramsky defines  models  arising from  {\em Geometry of Interaction} ({\em GoI}) {\em situations},
consisting of {\em history-free strategies}. He discusses $\mathcal{I}$, the model of {\em partial injections} and $\mathcal{P}$, its substructure consisting of {\em partial involutions}. In particular, S. Abramsky introduces notions of {\em reversible pattern-matching} ({\em bi-orthogonal}) {\em automata} as concrete devices for implementing such strategies. In this paper, we focus on the model  ${\cal P}$ of partial involutions, similar results can be obtained for the one of partial injections.

The model of partial involutions yields an {\em affine combinatory algebra}. This notion extends that of {\em strictly affine combinatory algebra}, introduced in Definition \ref{lca}, with a $!$ operation and extra combinators:

\begin{definition}[Affine Combinatory Algebra, \cite{Abr05}] \label{aca}
 An {\em affine combinatory algebra} ({\em ACA}),  $\mathcal{ A}= (A, \cdot, \bang)$ is
an applicative structure $(A, \cdot)$ with a unary (injective) operation $\bang $, and combinators $B, C,  I,K,  W, D, \delta,  F$ sati\-sfying the following
equations: for all $x,y,z \in A$,
\smallskip

\begin{tabular}{llll}
  $ Bxyz = x(yz) $ \ \ \ \ \ \ & $  Ix = x  $    \ \ \ \ \ \ &
$  Cxyz = (xz)y   $ \ \ \ \   \ \ &  $   K x  y = x  $
\\
$  Wx\bang y = x\bang y\bang y   $   \ \ \ \ \ \ &  $   \delta \bang x = \bang \bang x   $
 \ \ \  \ \ \ &         $  D!x = x  $    \ \ \  \ \ \   &  $  F\bang x\bang y = \bang (xy)  $.
\end{tabular}\smallskip
\end{definition}

Affine combinatory algebras are models of {\em affine combinatory logic}, which extends {\em strictly affine combinatory logic}, introduced in Definition \ref{sacl}, with !-operator and combinators $W$, $\delta$, $D$, and $F$.

Partial involutions are defined over a suitable language of moves, and they can be endowed with an ACA structure:

\begin{definition}[The Model of Partial Involutions $\cal P$]\label{cdot}\hfill
\\ (i)  $T_\Sigma$, the language   of {\em moves},  is defined by the signature $\Sigma_0$= \{$e$\}, $\Sigma_1$ = \{l,r\}, $\Sigma_2\ = \{<\ ,\ > \}$; terms  $r(x)$ are {\em output words}, while terms $l(x)$ are {\em input words} (often denoted simply by $rx$ and $lx$);\\
(ii) ${\cal P}$ is the set of {\em partial involutions} over $T_\Sigma$, \ie\ the set of all partial injective functions $ f: T_\Sigma \rightharpoonup T_\Sigma $ such that $f(u)=v \Leftrightarrow f(v)=u$; \\
(iii) the operation of  {\em replication} is defined by $! f \ =\ \{ (<t,u>,<t,v>)\ |\  t \in  T_\Sigma\ \wedge (u,v)\in f\}$;\\
(iv) the notion of {\em linear application} is defined by $f\cdot g =\ f_{rr}\cup (f_{rl};g;(f_{ll};g)^*;f_{lr})$,  where $f_{ij} \ =\ \{ (u,v)| (i(u),j(v))\in f\}$, for $i,j \in \{r,l\}$ (see Fig.~\ref{fig:lappdiag}), where ``$;$'' denotes postfix composition.
\end{definition}

\begin{figure}[!h]
\[\xymatrix{
\mathsf{in}\ar[r] & \bullet\ar[r]^{f_{rr}}\ar[d]_{f_{rl}} & \bullet\ar[r] & \mathsf{out}\\
                  & \bullet\ar@<.5ex>[r]^{g} & \bullet\ar@<.5ex>[l]^{f_{ll}}\ar[u]_{f_{lr}} &
}
\]
\caption{Flow of control in executing $f\cdot g$.}\label{fig:lappdiag}
\end{figure}
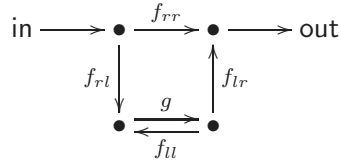

Following \cite{Abr05}, we make a slight abuse of notation and assume that $T_{\Sigma}$ contains  pattern variables for terms. The intended meaning will be clear from the context.
In the sequel, we will use the notation $u_1 \leftrightarrow v_1, \ldots , u_n \leftrightarrow v_n$, for $u_1, \ldots, u_n, v_1, \ldots , v_n \in T_{\Sigma}$, to denote the graph of the (finite)
partial involution $f$ defined by $\forall i. (f(u_i) = v_i \ \wedge \ f(v_i) = u_i)$. Again, following \cite{Abr05}, we will use the above notation in place of a more automata-like presentation of the partial involution.
\smallskip

\begin{proposition}[\cite{Abr05}, Th.5.1]
${\cal P}$ can be endowed with the structure of an {\em affine combinatory algebra}, $({\cal P}, \cdot, !)$, where combinators are defined by the following partial involutions:
\\
\begin{tabular}{lll}
$B$ & : &   $ r^3x \leftrightarrow lrx \ , \  l^2 x\leftrightarrow rlrx \ , \   rl^2 x \leftrightarrow r^2 lx $\\
$C$ & : &   $ l^2x \leftrightarrow r^2 lx \ , \  lrl x\leftrightarrow rlx \ , \  lr^2 x \leftrightarrow r^3 x $\\
$F$ & : &
$l  \langle  x, ry\rangle \leftrightarrow r^2   \langle  x,  y \rangle  \ , \  l  \langle  x, ly\rangle \leftrightarrow rl   \langle  x,  y \rangle$\\
$W$ & : &
$ r^2x \leftrightarrow lr^2 x \ , \  l^2  \langle x, y\rangle \leftrightarrow rl \langle lx, y\rangle  \ , \   lrl  \langle x, y\rangle \leftrightarrow rl \langle rx, y\rangle$\\
$I$ & : &  $lx \leftrightarrow rx$\\
$K$ & : &  $lx \leftrightarrow r^2x$\\
$\delta$ & : & $l  \langle \langle  x, y\rangle, z \rangle \leftrightarrow r   \langle  x, \langle y, z \rangle \rangle$\\
$D$ & : & $l  \langle e, x\rangle \leftrightarrow r x $.
\end{tabular}
\end{proposition}

In Section~\ref{qua} we focus on the strictly linear and affine part of the above combinatory algebra, \emph{i.e.} $({\cal P}, \cdot)$ together with combinators {\bf BCI} (and {\bf K}).

\subsection{From Principal Types to Involutions (and back)}

Our approach makes it possible to highlight a {\em duality} between {\em principal type schemes} and the interpretation of {\em strict and affine} combinators as involutions in $\cal P$. The following algorithm is a \emph{transform} which, given a principal type scheme, \ie\ a {\em global} representation of an object, yields for each type-variable (a component of) an involution:

\begin{definition} \label{Principal} Given a closed term $M$ of  $\lambda^A$-calculus such that $\Vdash_A M :\mu$, for each type variable $\alpha \in \mu$, the judgements ${\cal T}(\alpha, \mu) $ yield a pair in the graph of a
partial involution, if $\alpha$ occurs twice in $\mu$, or an element of $T_{\Sigma}$, if $\alpha$ occurs once in $\mu$: \medskip \\
${\cal T}(\alpha, \alpha)\ =\ \alpha$    \ \ \ \ \ \ \ \ \ \ \
${\cal T}(\alpha, \mu(\alpha) \rightarrow \nu(\alpha))\ =\ l({\cal T}(\alpha,  \mu(\alpha)))\ \leftrightarrow\ r({\cal T}(\alpha,  \nu(\alpha)))$ \medskip \\
${\cal T}(\alpha, \mu(\alpha) \rightarrow \nu)\ =\ l[{\cal T}(\alpha,  \mu(\alpha))] $ \ \ \
${\cal T}(\alpha, \mu \rightarrow \nu(\alpha))\ =\ r[{\cal T}(\alpha,  \nu(\alpha))]$  \medskip \\
where $r[x] = \begin{cases} rx_1 \leftrightarrow rx_2 & \mbox {if } x= x_1 \leftrightarrow x_2 \ \wedge \ x_1, x_2\in T_{\Sigma}
\\ rx & \mbox{\ \ otherwise}
\end{cases}$
\\  and similarly for $l[x]$. \medskip \\
We define the partial involution
$ f_{\mu} = \{ {\cal T} (\alpha, \mu)  \mid  \alpha \mbox{\ appears twice in } \mu \}  \ .$

\end{definition}

Vice versa, any partial involution interpreting a closed ${\bf {CL}^{A}}$-term $M$ induces  the corresponding  {\em principal type}, inverting the clauses in Definition \ref{Principal}. Notice that so doing we can derive a principal type scheme from any partial involution, not just those which are indeed interpretations of $\lambda$-terms. This remark will be crucial in addressing Abramsky's open question in Section \ref{openquestion}.

\begin{definition}
We denote by $\p{\ }_{\cal P}$ the interpretation of closed ${\bf CL^A}$-terms in $({\cal P}, \cdot)$.
\end{definition}

\begin{theorem}\label{ven}
Given a closed  term of ${\bf {CL}^{A}}$, say $M$, the partial involution interpreting $M$,  namely $\p{M}_{\cal P}$, can be read off the principal type scheme of  $(M)_{\lambda^{A}}$, \ie\
$\Vdash (M)_{\lambda^{A}}: \mu$ if and only if $ \p{M}_{\cal P}= f_{\mu}$.
\end{theorem}
\begin{proof} (Sketch)
By induction on the structure of ${\bf {CL}^{A}}$-terms. One can easily check that the thesis holds for combinators B, C, I, K. The inductive step amounts to showing that the notion of application in ${\cal P}$ corresponds to computing the  principal type scheme of the application, \emph{i.e.}, for $MN$ closed
${\bf {CL}^{A}}$-term, if $\Vdash M:\mu$, $\Vdash N:\tau$, $\mathit{TVar}(\mu) \cap \mathit{TVar}(\tau) =\emptyset$, $U'=MGU(\mu, \alpha \rightarrow \beta)$,
$U=  MGU(U'(\alpha), \tau)$, $\alpha, \beta$ fresh,  then
$f_{U\circ U' (\beta)} = f_{\mu} \cdot f_{\tau}$. This latter fact can be proved by chasing, along the control flow diagram in the definition of application, the behaviour of the MGU.
\end{proof}

We are finally in the position of justifying the claims, in the introduction, that our analysis unveils three conceptually independent, but ultimately equivalent, accounts  of {\em application} in the $\lambda$-calculus: {\em $\beta$-reduction}, the GoI application of involutions based on symmetric feedback/Girard's {\em Execution Formula}, and {\em unification} of principal types. In effect, computing the partial involutions $\p{M}_{\cal P} \cdot \p{N}_{\cal P}$, according to Definition \ref{cdot},  amounts by Theorem \ref{ven} to {\em unifying} the left-hand side of the principal type of $M$ with the principal type of $N$, thus computing the principal type of $MN$. Using Definition \ref{Principal} we can finally read off from this type scheme the partial involution $\p{MN}_{\cal P}$.

The following theorem concludes our model theoretic analysis:
\begin{theorem}\label{Nocomb}
\hfill
\\ (i) The strictly linear combinatory algebra of partial involutions
$(\cal P, \cdot)$ is a strictly linear  $\lambda$-algebra, albeit not a strictly linear combinatory $\lambda$-model.
\\ (ii) The strictly affine combinatory algebra of partial involutions
$(\cal P, \cdot)$ is not a strictly affine $\lambda$-algebra.
\end{theorem}

\begin{proof}\hfill
\\ (i) All the equations in Theorem \ref{th:eqforSALA} have been verified.
In order to avoid errors we used the Erlang~\cite{Erlang1,Erlang2} program described in Section~\ref{erlang}.
\item The proof that there does not exist a term which behaves as a selector combinator, namely that $\epsilon$ does not exist, follows from Lemma \ref{nocomb} below. The combination of items (i) and (ii) of Lemma~\ref{nocomb} below contradicts the unique selection property of $\epsilon$, namely we exhibit two objects,  \ie\ $\emptyset$ and $\{ l\alpha \leftrightarrow l \alpha \}$, which have the same empty applicative behaviour,  but for any $E$ which satisfies $\forall x,y.\ E \cdot x \cdot y= x \cdot y$, we have $E \cdot \emptyset \neq E \cdot \{ l\alpha \leftrightarrow l \alpha\}$. Consider first the terms $ X = \{ r\alpha \leftrightarrow l\alpha ,  l\beta \leftrightarrow l\beta\}$ and $Y= \{ \alpha \leftrightarrow \beta\}$, with $\alpha\neq \beta$. Clearly we have $X \cdot Y = \{\alpha \leftrightarrow \alpha\}$.
But $E \cdot X =  E \cdot \{ r\alpha \leftrightarrow l\alpha\} \cup E
\cdot \{ l\beta \leftrightarrow l\beta\}$, by Lemma~\ref{nocomb} (ii). Now  $E \cdot \{ r\alpha \leftrightarrow l\alpha\} \cdot Y = \{ r\alpha \leftrightarrow l\alpha\} \cdot Y  = \emptyset $. Hence $E \cdot \{ l\beta \leftrightarrow l\beta\} \neq \emptyset$, since $E\cdot X\cdot Y =X\cdot Y \neq \emptyset$.\smallskip
\\ (ii) We have that $(BBK)_{\lambda^A} = (BKK)_{\lambda^A}$, but  $\p{BBK }_{\cal P} \neq \p{BKK}_{\cal P}$. Namely, $\p{BBK }_{\cal P} = \p{ \lambda^* xyz.Kx(yz) }_{\cal P} = \{lx \leftrightarrow r^3x, rl^2x\leftrightarrow r^2lx\}$ while $\p{BKK }_{\cal P}= \p{\lambda^* xyz.x }_{\cal P} =  \{lx \leftrightarrow r^3x\}.$
\end{proof}

\begin{lemma}\label{nocomb}
Assume that there exists $E \in {\cal P}$ such that $\forall x,y. \ E \cdot x \cdot y\ =\ x \cdot y$, then
\\
(i)   $E_{rr} = \emptyset$, and hence $E \cdot \emptyset = \emptyset$;
\\ (ii)   $E_{ll} = \emptyset$,  and hence $E$ has an ``additive'' applicative behaviour, namely $E \cdot (A \cup B)= (E\cdot A) \cup (E \cdot B)$.
\end{lemma}

\begin{proof} \hfill
\\ (i) We proceed in stages.
\begin{itemize}

\item $r\alpha \leftrightarrow r\beta \notin E_{rr}$,  for any $\alpha,\beta$,  otherwise $\alpha \leftrightarrow \beta \in E \cdot \emptyset \cdot  X =  \emptyset \cdot X =\emptyset $, contradiction.

\item  $r\alpha \leftrightarrow l\beta \notin E_{rr}$, for any $\alpha,\beta$,  otherwise let
$A = \{r\alpha \leftrightarrow l\delta, r\delta \leftrightarrow l\alpha\}$  with $\delta$ and $\beta$ not unifiable. Then $A \cdot \{\beta \leftrightarrow \beta\}\ =\ \emptyset$ , but $E \cdot A \cdot \{\beta \leftrightarrow \beta\} =  \{\alpha\leftrightarrow \alpha\}$. Contradiction.

\item  $l\alpha \leftrightarrow l\beta \notin E_{rr}$, for any $\alpha \neq \beta$,
otherwise let $A = \{r\alpha \leftrightarrow l\alpha\}$, then $A \cdot \{\alpha \leftrightarrow \alpha\} = \{\alpha \leftrightarrow \alpha\}$.
Now, since $E$ is a selector, $E \cdot A \cdot \{\alpha \leftrightarrow \alpha\}= \{\alpha \leftrightarrow \alpha \}$,  there must occur in $E\cdot A$ a term $r\alpha \leftrightarrow l\rho$ such that $\rho$ can unify with the l.h.s. of the only pair in the argument, \ie\ $\alpha = \rho$, but then $E \cdot A$ would no longer be a non-ambiguous reversible relation, because by assumption $ l\alpha \leftrightarrow l\beta\in E \cdot A$. Contradiction.
\item Similar arguments can be used to rule out the remaining cases, \ie\ $l\alpha \leftrightarrow l\alpha\in E_{rr}$ and the case where one of the components is garbage, \ie\ it has no functional effect.
\end{itemize}
(ii) We proceed in stages.
\begin{itemize}
\item From the very definition of application in $\cal P$, we have immediately that, if  $A_{ll} = \emptyset$, then $A$ has an ``additive'' behaviour under application, because it calls the argument only once.
\item $E_{ll} = \emptyset$ because for all $l(\alpha)$ either  $ll\alpha \leftrightarrow rr\alpha \in E$ or $lr\alpha \leftrightarrow rr\alpha \in E$ and $rl\alpha \leftrightarrow ll\alpha \in E$, and hence there are no $l\alpha \longrightarrow l\beta\in E_{ll}$, since $E$ is an involution and the rewrite rules are deterministic.
To see the above, first notice that $r\alpha \leftrightarrow l\alpha\ \in \ E \cdot \{ r\alpha \leftrightarrow l\alpha\} $ for all $\alpha$,  otherwise, checking the control-flow diagram, one can easily see that we could not have that  $ E \cdot \{ r\alpha \leftrightarrow l\alpha\} \cdot \{\alpha \leftrightarrow \alpha\} =  \{\alpha \leftrightarrow \alpha\}$.
But now, again with just a little case analysis on the control-flow diagram, one can see that there are only two alternatives in $E_{rl}$ and $E_{lr}$, which give rise to the cases above.
\\ The only case left is $le \leftrightarrow le \in E_{ll}$. But then we would have that
\\ $\alpha \leftrightarrow \alpha \in (E \cdot \{r\alpha \leftrightarrow l\alpha\}\cdot \{\alpha \leftrightarrow e\})$, but
 $\{r\alpha \leftrightarrow l\alpha\}\cdot \{\alpha \leftrightarrow e\} = \emptyset$, contradiction.
\noindent
\\ Hence we have that
$\{ r\alpha \leftrightarrow l\alpha ,  r\delta \leftrightarrow l\delta\} \cdot \{\alpha \leftrightarrow \beta , \gamma \leftrightarrow \delta \} =  \emptyset$, \\
but
 $E \cdot \{ r\alpha \leftrightarrow l\alpha , r\delta \leftrightarrow l\delta\} \cdot \{\alpha \leftrightarrow \beta , \gamma \leftrightarrow \delta \} \supseteq \{ \alpha \leftrightarrow \beta\}$, contradiction.
\end{itemize}
\end{proof}

\subsubsection{Abramsky's Question.}\label{openquestion}
In \cite{Abr05}, S. Abramsky raised the question: ``Characterize those partial involutions which arise from {\em bi-orthogonal pattern-matching automata}, or alternatively, those which arise as denotations of combinators''.

Theorem~\ref{ven}  suggests an answer to the above question for the strictly affine fragment, \ie\ without the operator $<\ ,\ >$ in the language of partial involutions.
The first issue to address is how to present partial involutions. To this end we consider the language $T_{\Sigma^X}$, which is the initial term algebra over the signature $\Sigma^X$ for ${\Sigma^X_0} \equiv X$, where  $X$ is a set of variables, and ${\Sigma^X_1} = \{l,r\}$. Sets of pairs in $T_{\Sigma^X}$ denote schemata of pairs over $T_{\Sigma\setminus \Sigma_2}$, \ie\ partial involutions in $\cal P$.
As pointed out in the previous section, given a partial involution defined by a finite collection of pairs in $T_{\Sigma^X}$, say $H\equiv \{u_i \leftrightarrow v_i\}_{i\in I}$ for $u_i,v_i \in {\Sigma^X}$, we  can synthesize a type $\tau_H$ from $H$ by gradually specifying its tree-like format. Finally we check whether $\tau_H$ is the principal type of a strictly linear term. We proceed as follows. Each pair in $H$ will denote two leaves in the type $\tau_H$, tagged with the same type variable. The sequence of $l$'s and $r$'s, appearing in the prefix before a variable in a pair $u_i,v_i$, denotes the path along the tree of the type $\tau_H$, which is under formation, where the type variable will occur. A fresh type variable is used for each different pair.  At the end of this process we might not yet have obtained a complete type. Some leaves in the tree might not be tagged yet, these arise in correspondence of vacuous abstractions. We tag each such node with a new fresh type variable. $H$ is  finite otherwise we end up with an infinite type, which cannot be the principal type of a {\em finite} combinator. The type $\tau_H$ thus obtained has the property that each type variable occurs at most twice in it. Potentially it is  a  principal type.

The type $\tau_H$ is indeed a principal type of a closable $\lambda$-term (\ie\ a term which reduces to a closed term) if and only if it is an {\em implication tautology} in minimal logic. This can be effectively checked in polynomial-space~\cite{Statman79}.

To complete the argument we need to show that if the type $\tau_H$ is inhabited it is indeed inhabited by a term for which it is the  principal type.
\begin{proposition}\label{auspicio}
If $\mu$ is a type where each type variable occurs at most twice and it is inhabited by the  closed term $M$, then there exists $N$ such that $\Gamma \Vdash_L N: \mu$ and $N=_{\beta_L\eta}M$.
\end{proposition}
\begin{proof}(Sketch)
If $M$ is a closed term, then there exists $\nu$ such that $\Gamma \Vdash_L M: \nu$. The variables in $\Gamma$ will be eventually erased. If  $M$ inhabits $\mu$, then by Theorem~\ref{theorem_principal} there exists a substitution $U$ such that $U(\nu)=\mu$. For each variable which is substituted by $U$, say $\alpha$, two cases can arise, either $\alpha$ occurs twice or once.  In the first case we will replace the term variable, say $x$, in $M$ in whose type $\alpha$  occurs, which must exist, by a suitable long-$\eta$-expansion of $x$. This long $\eta$-expansion can always be carried out because the typed $\eta$-expansion rule is a derivable rule in the typing system.
\\
\noindent In case the type variable $\alpha$ occurs only once in $M$,  there is a subterm of $M$ which is embedded in a vacuous abstraction.  The term $N$ is obtained by nesting that  subterm  with a new vacuous $\lambda$-abstraction applied to a long-$\eta$-expansion of the variable vacuously abstracted in $M$.

\noindent Here are two examples. From
\noindent $H_1 = \{ lllx \leftrightarrow rllx, llrx\leftrightarrow lrx,rrx \leftrightarrow rlx \}$ we can synthesize the type $((\alpha \rightarrow \beta)\rightarrow \gamma)\rightarrow (\alpha \rightarrow \beta)\rightarrow \gamma$. The identity, $\lambda x.x$, inhabits this type, but the type is not the principal type of the identity.  It is  instead the principal type of an $\eta$-expansion of the identity, namely $\lambda xy. x(\lambda z.yz)$.\\
\noindent From $H_2 = \{ lllx \leftrightarrow lrrx, llrx\leftrightarrow lrlx,lrrx \leftrightarrow rrrx \}$ we can synthesize the type $((\alpha \rightarrow \beta)\rightarrow (\beta \rightarrow \alpha))\rightarrow \gamma \rightarrow \gamma$. The term $\lambda yx.x$ inhabits this type which is the principal type of its $\beta$-expansion $\lambda yx. (\lambda w. x)(\lambda zw.yzw)$.

\end{proof}

So we can finally state the result:

\begin{theorem}
In the {\em strictly affine} case, the denotations of combinators in $\cal P$ are precisely the partial involutions from which we can synthesize, according to the procedure outlined above, a  principal type scheme which is a tautology in minimal logic.
\end{theorem}
\begin{proof} Use Proposition~\ref{auspicio} above in one direction, and Definition~\ref{Principal} and Theorem~\ref{ven} in the opposite direction.
\end{proof}

The above is a satisfactory characterisation because it is conceptually independent both from $\lambda$-terms and from involutions.

\section{Automating Abramsky's Linear Application}\label{erlang}
Since the manual verification of complicated equations like those appearing in Theorem~\ref{th:eqforSALA} is a daunting and error-prone task, we developed an Erlang~\cite{Erlang1,Erlang2} program to automatize it (see the Appendix for the details). The main components of our program are related to the implementation of the \emph{Strictly Affine Abstraction Operator} (see Definition~\ref{op}) and of the linear application operator $f\cdot g$ (see Definition~\ref{cdot}) introduced by S. Abramsky in~\cite{Abr05}.

There are several reasons behind the choice of Erlang: expressive pattern matching mechanisms, avoidance of side effects thanks to one-time assignment variables, powerful libraries for smooth handling of lists, tuples etc. However, other functional languages can be an effective and viable choice as well.

Rules can be written as, \eg\ $l(x)\leftrightarrow r(x)$ (i.e., the representation of combinator \textbf{I} as a partial involution). For convenience, we allow the user to avoid specifying parentheses where their occurrence can be automatically inferred like in, e.g., $lx\leftrightarrow rx$. Moreover, we allow the user to write $\langle x,y\rangle$ instead of $p(x,y)$.

For instance, starting from a string representing the three rules of combinator \textbf{B} (i.e., \verb|"rrrX<->lrX, llX<->rlrX, rllX<->rrlX"|\footnote{In our implementation we actually use capital letters or strings with initial capital letter to denote variables, according to Erlang conventions.}), we obtain the following internal representation (each rule with \texttt{<->} yields two internal rules corresponding to the two possible directions of rewriting, since it is more convenient for coding purposes):
\begin{verbatim}
[{{r,{r,{r,{var,"X"}}}},{l,{r,{var,"X"}}}},
 {{l,{r,{var,"X"}}},{r,{r,{r,{var,"X"}}}}},
 {{l,{l,{var,"X"}}},{r,{l,{r,{var,"X"}}}}},
 {{r,{l,{r,{var,"X"}}}},{l,{l,{var,"X"}}}},
 {{r,{l,{l,{var,"X"}}}},{r,{r,{l,{var,"X"}}}}},
 {{r,{r,{l,{var,"X"}}}},{r,{l,{l,{var,"X"}}}}}
]
\end{verbatim}

Indeed, we can compute the composition $f;g$ (\texttt{compose(F,G)}) of two involutions $f$ and $g$ as the set of rules $\{(R_1,R_2)\mid R_1=s(F_1),\ R_2=s(G_2),\ (F_1,F_2)\in f, (G_1, G_2)\in g, \mathrm{and}\ s=m.g.u.(F_2,G_1)\}$, where m.g.u. stands for \emph{most general unifier} which can be implemented following Robinson's unification algorithm~\cite{Rob65}.
There is only a subtle issue to take into consideration, namely, unification may not work correctly if the sets of variables of $f$ and $g$ are not disjoint. Hence, we preventively rename variables of $f$ if this is not the case in the computation of $f;g$.
Once the implementation of $f;g$ is completed, it is trivial to define $f\cdot g$, unfolding its definition in terms of the composition operator and calculating $f_{rr}$ (\texttt{extract(F,r,r)}), $f_{rl}$ (\texttt{extract(F,r,l)}), $f_{ll}$ (\texttt{extract(F,l,l)}), $f_{lr}$ (\texttt{extract(F,l,r)}), exploiting in \texttt{extract} the powerful pattern matching features of Erlang. This approach has been used to verify all the equations appearing in this paper.

Let us see, as an example, how the verification of equation $\lambda^*xyz.C(C(BBx)$ $y)z=\lambda^*xyz.Cx(yz)$ from Theorem~\ref{th:eqforSALA} is carried out.
We first use \texttt{leex} and \texttt{yecc} (\emph{i.e.} the Erlang versions of Lex and Yacc) to build a lexical analyzer and a parser for the language of $\lambda$-terms, yielding an internal representation of the two $\lambda$-abstractions. Then, we apply the implementation of the \emph{Strictly Affine Abstraction Operator} to yield the following expressions:
\begin{enumerate}
\item $((C((BC)((B(BB))((B(BC))((C((BB)((BC)((B(BB))I))))I)))))I)$
\item $((C((BB)((BB)((BC)I))))((C((BB)I))I))$
\end{enumerate}
corresponding, respectively, to $\lambda^*xyz.C(C(BBx)y)z$ and to $\lambda^*xyz.Cx(yz)$.
Afterwards, we use again \texttt{leex} and \texttt{yecc} to process the language of terms of partial involutions: $T::=e\mid l(T)\mid r(T)\mid p(T,T)\mid x$, where $x$ represents a variable. We allow the user to input involution rules, using \verb|<->| to denote rewritings, as one would do with pencil and paper, in order to encode the combinators according to their representations (given in~\cite{Abr05}) as partial involutions.
After the parsing process, each combinator is internally represented as a list of rewriting clauses and it is only a matter of applying the application operation between partial involutions (\emph{i.e.} $f\cdot g$) in order to check that the combinator expressions are indeed equal, \emph{i.e.}, they generate the same involution.

For further details and for the implementation of the replication operator ($!$) we address the interested reader to the Appendix.

\section{Final Remarks and Directions for Future Work}\label{fin}
In this paper, we have analysed from the point of view of the model theory of  $\lambda$-calculus the combinatory algebras arising in \cite{Abr05}, consisting of partial involutions expressing the behaviour of reversible pattern-matching automata. We have shown that the last step of  ``Abramsky's Programme'', taking from a linear combinatory algebra to a $\lambda$-algebra is not immediate. Only in the strictly linear case we have a $\lambda$-algebra, but already in the strictly affine case the model of partial involutions cannot be immediately turned into a $\lambda$-algebra. To check the necessary equations, we have implemented in the language Erlang the application between partial involutions. In \cite{CDGHLS} we show how to define a suitable quotient to achieve a $\lambda$-algebra from the strictly affine combinatory algebra of partial involutions.

A key tool in analysing partial involutions interpreting combinators is the duality between involutions and the principal types, w.r.t. a suitable  type discipline, of the combinators expressed in the $\lambda$-calculus. This alternate characterization of the partial involutions yields also an answer to a question raised by S. Abramsky in \cite{Abr05}.
This has been  worked out in detail in the strictly affine case.

In~\cite{LC18,CDGHLS}, we generalize all the results in this paper to the full affine case. Namely, we
 define a  notion of affine $\lambda$-calculus, the $\lambda^!$-calculus, including terms $!M$, and two kinds of $\lambda$-abstractions, strictly affine and
!-abstraction, which abstract general $\lambda$-terms. The Abstraction Operation in Definition~\ref{op} and the Abstraction Theorem~\ref{absth}  can then be extended to the full affine case. A counterpart to Theorem~\ref{th:eqforSALA} can be proved for full affine $\lambda$-algebras.  Extending the type system to the $\lambda^!$-calculus requires
special care: a suitable intersection type system is needed, in the line of \cite{DGHL,DGL13}, where a $!_u$ type operator needs to be added. Then, the connection between partial involutions in the algebra $\cal P$ and principal type schemes, extending Theorem~\ref{ven}, can be stated. This latter result provides an answer to the full version of Abramsky's question along the lines of the procedure described in the present paper, for the strictly affine case.

An interesting problem to address is to characterize the fine theory of $\cal P$.

Finally, another challenging direction for future work could be to explore the potential connections with a vast part of the literature devoted to provide an implementation of GoI (see, e.g., \cite{Gonthier92,Mackie95}) and, in particular, reversible approaches (e.g., \cite{Danos99}).

\newpage
\begin{appendices}
\setcounter{chapter}{1}

\section{Implementing \textsf{LApp} in Erlang}\label{web-app}
In this Appendix we will outline in full details our Erlang~\cite{Erlang1,Erlang2} implementation of Abramsky's \textsf{LApp} operator. We recall from~\ref{erlang} that we built a lexical analyzer and a syntactical parser for the language of partial involutions in order to allow the user to write rewriting rules in a user-friendly way, namely, as strings like the following one (for combinator \textbf{B}):
\begin{verbatim}
"rrrX<->lrX, llX<->rlrX, rllX<->rrlX"
\end{verbatim}
Once parsed, the previous string will yield an internal representation as a list  of pairs: each pair is a directed (\texttt{->}) rewriting rule (hence, there are twice the rules for each clause in the original string, i.e., one for each rewriting direction). Thus, combinator \textbf{B} will be represented internally by the following list of pairs:
\begin{verbatim}
[{{r,{r,{r,{var,"X"}}}},{l,{r,{var,"X"}}}},
 {{l,{r,{var,"X"}}},{r,{r,{r,{var,"X"}}}}},
 {{l,{l,{var,"X"}}},{r,{l,{r,{var,"X"}}}}},
 {{r,{l,{r,{var,"X"}}}},{l,{l,{var,"X"}}}},
 {{r,{l,{l,{var,"X"}}}},{r,{r,{l,{var,"X"}}}}},
 {{r,{r,{l,{var,"X"}}}},{r,{l,{l,{var,"X"}}}}}]
\end{verbatim}
In this appendix all the functions we will introduce will operate on the internal representation (since it is more convenient for coding purposes); hence, we will not describe the lexical analyzer and the parser which are automatically generated from the grammar specification of the language of partial involutions, using \texttt{leex} and \texttt{yecc} (i.e., the Erlang versions of the well known tools Lex and Yacc).

\subsection{Auxiliary functions}
First of all, we begin by introducing some basic auxiliary functions allowing us to deal smoothly with tasks such as eliminating duplicates from a list, extracting variables occurring in a given term/list, generating a \emph{fresh} variable etc.

\subsubsection{Removing duplicates from a list}\label{app:subsubsec:dd}
The recursive definition is self-explanatory once we notice that \texttt{member(H,T)} is the library function returning true if and only if \texttt{H} is a member of \texttt{T}.
\begin{lstlisting}
dd(L) ->
  case L of
    [] -> [];
    [H|T] ->
      Tail=dd(T),
      Exists=lists:member(H,T),
      if
        Exists -> Tail;
        true -> [H]++Tail
      end
  end.
\end{lstlisting}
The \texttt{true} clause corresponds to the \emph{else} branch of the conditional command in common imperative languages.

\subsubsection{Computing the list of variables occurring in terms and rewriting rule sets}\label{app:subsubsec:vars}
Given a term \texttt{T} built starting from variables and constructors of the language of partial involutions (namely, $\epsilon$, $r$, $l$, and $p$), the following function definition allows us to recursively compute the list of variables occurring into \texttt{T}:

\begin{lstlisting}
vars(T) ->
  case T of
    e -> [];
    {var, X} -> [X];
    {l,U} -> vars(U);
    {r,U} -> vars(U);
    {p,U1,U2} -> V1=vars(U1), V2=vars(U2),
                 if
                   V1==[] -> V2;
                   V2==[] -> V1;
                   true -> dd(V1++V2)
                 end
  end.
\end{lstlisting}
Notice the use of the \emph{dd} function defined in Section~\ref{app:subsubsec:dd}, in order to avoid duplicates in the \texttt{p} constructor case.

We then use \texttt{vars} in order to specify the function \texttt{ruleListVars} which returns the list of variables occurring in the list of rewriting rules \texttt{L} passed as a parameter:
\begin{lstlisting}
ruleListVars(L) ->
  case L of
    [] -> [];
    [{R1,R2}|Tail] -> dd(vars(R1)++vars(R2)++ruleListVars(Tail))
  end.
\end{lstlisting}

\subsubsection{Generating fresh variables to avoid variables clashes}\label{app:subsubsec:fresh}
In order to avoid clashes with variables of different rewriting rule sets having the same name, we need a fresh renaming mechanism. The following function definition accepts a list of variables \texttt{L} as a parameter and returns a variable not occurring in \texttt{L}:
\begin{lstlisting}
fresh(L) ->
  case L of
    [] -> "_X1";
    [Head|Tail] -> Var=fresh(Tail),
                   Len=length(Head),

                   if
                     Len>=2 -> Prefix=string:substr(Head, 1, 2),
                               if
                                 Prefix=="_X" -> {HeadId,_}=string:to_integer(string:substr(Head, 3, Len)),
                                 if
                                   HeadId==error -> Var;
                                   true -> {VarId,_}=string:to_integer(string:substr(Var, 3, length(Var))),
                                           if
                                             VarId>HeadId -> Var;
                                             true -> lists:flatten(io_lib:format("_X~p", [HeadId+1]))
                                           end
                                 end;
                                 true -> Var
                               end;
                     true -> Var
                   end
  end.
\end{lstlisting}
Automatically generated variables have names of the following shape: \texttt{\_Xn}, where \texttt{n} is an integer. Thus, they can be easily identified.

The \texttt{fresh} function is used in the definition of \texttt{separateVars} which returns a list of substitutions in order to separate variables occurring in \texttt{Vars1} from those occurring in \texttt{Vars2}:
\begin{lstlisting}
separateVars(Vars1,Vars2) ->
  case Vars1 of
    [] -> [];
    [Var|Tail] -> Check=lists:member(Var,Vars2),
                  if
                    Check -> NewVar=fresh(Vars1++Vars2),
                             [{Var,{var, NewVar}}]++separateVars(Tail,[Var,NewVar]++Vars2);
                    true -> separateVars(Tail,[Var|Vars2])
                  end
  end.
\end{lstlisting}
The attentive reader should notice that the new variable \texttt{NewVar} is chosen fresh w.r.t. both \texttt{Vars1} and \texttt{Vars2}. Moreover, we append both \texttt{Var} and \texttt{NewVar} to the second argument in the recursive calls; hence, the whole mechanism will work even after previous \emph{fresh renamings}.

\subsubsection{Substitutions}\label{app:subsubsec:subst}
In this section we deal with various kinds of substitutions. We start with a function returning \texttt{T[Y/X]}, i.e., the result of substituting \texttt{Y} for \texttt{X} in \texttt{T}:
\begin{lstlisting}
subTerm(X,Y,T) ->
  case T of
    e -> e;
    {var,V} -> if
                 X == V -> Y;
                 true -> {var, V}
               end;
    {l,U} -> {l,subTerm(X,Y,U)};
    {r,U} -> {r,subTerm(X,Y,U)};
    {p,P1,P2} -> {p,subTerm(X,Y,P1),subTerm(X,Y,P2)}
  end.
\end{lstlisting}
Then, \texttt{subTerm} is used to implement two other types of multiple substitutions, namely, \texttt{subListTerm} where a list \texttt{L} of substitutions \texttt{[T1/X1]},\ldots,\texttt{[Tn/Xn]} is applied to a term \texttt{U}, yielding as a result \texttt{(\ldots(U[T1/X1])\ldots[Tn/Xn])}:
\begin{lstlisting}
subListTerm(L,U) ->
  case L of
    [] -> U;
    [{X,Y} | Tail] -> subListTerm(Tail,subTerm(X,Y,U))
  end.
\end{lstlisting}
and \texttt{subList} which is a function allowing one to substitute \texttt{Y} for \texttt{X} in the codomain of a list \texttt{L} of substitutions:
\begin{lstlisting}
subList(X,Y,L) ->
  case L of
    [] -> [];
    [{Z,U}|T] -> [{Z,subTerm(X,Y,U)} | subList(X,Y,T)]
  end.
\end{lstlisting}
The last form of substitution we need is the one implemented by \texttt{subListRuleset} (built on top of the previously defined \texttt{subListTerm}):
\begin{lstlisting}
subListRuleset(Subst,Ruleset) ->
  case Ruleset of
    [] -> [];
    [{R1,R2}|Tail] -> [{subListTerm(Subst,R1),subListTerm(Subst,R2)}]++subListRuleset(Subst,Tail)
  end.
\end{lstlisting}
in this case we are applying a list of substitutions represented by the parameter \texttt{Subst} to all the terms occurring in the rewriting rules of the \texttt{Ruleset} parameter.

\subsection{Robinson's unification algorithm}\label{app:subsec:robinson}
The unification algorithm, originally conceived by Robinson, is essential in order to apply in sequence rewriting rules:
\begin{lstlisting}
unify(X,Y,L) ->
  case X of
    e -> case Y of
           e -> {ok,L};
           {var,_} -> unify(Y,X,L);
           _ -> {fail,[]}
         end;
    {var,V} -> Check=not(lists:member(V,vars(Y))),
               if
                 Check -> {ok,dd(lists:append(subList(V,Y,L),[{V,Y}]))};
                 true -> {fail,[]}
               end;
    {l,U} -> case Y of
               {var,_} -> unify(Y,X,L);
               {l,V} -> unify(U,V,L);
               _ -> {fail,[]}
             end;
    {r,U} -> case Y of
               {var,_} -> unify(Y,X,L);
               {r,V} -> unify(U,V,L);
               _ -> {fail,[]}
             end;
    {p,P1,P2} -> case Y of
                   {var,_} -> unify(Y,X,L);
                   {p,Q1,Q2} -> {Flag1,R1}=unify(P1,Q1,L),
                                if
                                  (Flag1==ok) -> {Flag2,R2}=unify(subListTerm(R1,P2),subListTerm(R1,Q2),R1),
                                                 if
                                                   (Flag2==ok) -> {ok,R2};
                                                   true -> {fail,[]}
                                                 end;
                                  true -> {fail,[]}
                                end;
                   _ -> {fail,[]}
                 end
  end.
\end{lstlisting}
The intended meaning of \texttt{unify(X,Y,L)} is that the list \texttt{L} will be enriched by the substitutions needed to unify \texttt{X} and \texttt{Y}. Hence, if we start with the empty list, we will obtain the m.g.u. of \texttt{X} and \texttt{Y}. More precisely, the returned value is either a pair \texttt{\{ok,Mgu\}} in case of success, or a pair \texttt{\{fail,[]\}} if \texttt{X} and \texttt{Y} are not unifiable.

\subsection{Implementing \textsf{LApp}}\label{app:subsec:lapp}
As originally introduced by Abramsky, the definition of $\mathsf{LApp}(f,g)$ is given by $f_{rr}\cup f_{rl};g;(f_{ll};g)^*;f_{lr}$, where $f_{ij}=\{(u,v)\mid (i(u),j(v))\in f\}$ for $i,j\in\{l,r\}$. Hence, the ``flow of control'' of \textsf{LApp} can be graphically represented by the following diagram:
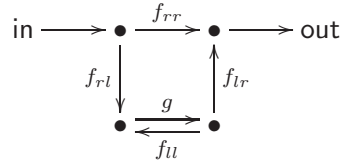
\[\xymatrix{
\mathsf{in}\ar[r] & \bullet\ar[r]^{f_{rr}}\ar[d]_{f_{rl}} & \bullet\ar[r] & \mathsf{out}\\
                  & \bullet\ar@<.5ex>[r]^{g} & \bullet\ar@<.5ex>[l]^{f_{ll}}\ar[u]_{f_{lr}} &
}
\]
Hence, we must begin programming a function named \texttt{extract}, being able to deduce l- and r- rewriting rules from \texttt{L}, according to \texttt{Op1} and \texttt{Op2}:
\begin{lstlisting}
extract(L,Op1,Op2) ->
  case L of
    [] -> [];
    [{e,_}|T] -> extract(T,Op1,Op2);
    [{_,e}|T] -> extract(T,Op1,Op2);
    [{{p,_,_},_}|T] -> extract(T,Op1,Op2);
    [{_,{p,_,_}}|T] -> extract(T,Op1,Op2);
    [{T1,T2}|T] -> {O1,S1}=T1,
                   {O2,S2}=T2,
                   if
                     (O1==Op1) and (O2==Op2) -> [{S1,S2} | extract(T,Op1,Op2)];
                     true -> extract(T,Op1,Op2)
                   end
  end.
\end{lstlisting}
Thus, if \texttt{F} represents a partial involution, then \texttt{extract(F,r,l)} will compute \texttt{F}$_{rl}$.

Then, we have the \emph{core} function \texttt{composeRuleList} which composes rule \texttt{R1}$\rightarrow$\texttt{R2} with all the rules in \texttt{L} (exploiting the unification and substitution functions defined in the previous sections):
\begin{lstlisting}
composeRuleList(R1,R2,L) ->
  case L of
    [] -> [];
    [{S1,S2}|T] -> {ExitStatus,MGU}=unify(R2,S1,[]),
                   if
                     (ExitStatus==ok) -> [{subListTerm(MGU,R1),subListTerm(MGU,S2)} | composeRuleList(R1,R2,T)];
                     true -> composeRuleList(R1,R2,T)
                   end
  end.
\end{lstlisting}

In order to avoid possible variable names clashes, the \texttt{alpha} function defined below replaces all variables in Ruleset1 which also occur in Ruleset2 with freshly generated ones:
\begin{lstlisting}
alpha(Ruleset1,Ruleset2) ->
  Vars1=ruleListVars(Ruleset1),
  Vars2=ruleListVars(Ruleset2),
  FreshSubst=separateVars(Vars1,Vars2),
  subListRuleset(FreshSubst,Ruleset1).
\end{lstlisting}

\texttt{alpha} is fruitfully used in the definition of \texttt{compose} which computes all possible chainings between rewriting rules of \texttt{L1} and \texttt{L2}:
\begin{lstlisting}
compose(L1,L2) ->
  L1_Fresh=alpha(L1,L2),
  compose_fresh(L1_Fresh,L2).

compose_fresh(L1_Fresh,L2) ->
  case L1_Fresh of
    [] -> [];
    [H1|T1] -> {R1,R2}=H1,
               composeRuleList(R1,R2,L2)++compose_fresh(T1,L2)
  end.
\end{lstlisting}

The funtion \texttt{star} will capitalize on the definition of \texttt{compose}, in order to implement the computation of $\mathtt{H};(\mathtt{F};\mathtt{G})^*$:
\begin{lstlisting}
star(H,F,G) ->
  S=compose(H,F),
  if
    S==[] -> H;
    true -> T=compose(S,G),
            if
              T==[] -> H;
              true -> H++star(T,F,G)
            end
  end.
\end{lstlisting}

So far, the definition of our implementation of \textsf{LApp}, according to Abramsky's specification, is straightforward:
\begin{lstlisting}
lapp(F,G) ->
    FRR=extract(F,r,r),
    FRL=extract(F,r,l),
    FLL=extract(F,l,l),
    FLR=extract(F,l,r),
    FRL_G=compose(FRL,G),
    FRL_G_STAR=star(FRL_G,FLL,G),
    FRR++compose(FRL_G_STAR,FLR).
\end{lstlisting}
In order to avoid an excessive nesting of \texttt{lapp} applications when dealing with complicate expressions, we also implemented the following functions:
\begin{lstlisting}
chainApp(Rulesets) ->
  Rev=lists:reverse(Rulesets),
  chain(Rev).

chain(Rulesets) ->
  case Rulesets of
    [] -> [];
    [R] -> R;
    [R|Tail] -> Rec=chain(Tail),
                lapp(Rec,R)
  end.
\end{lstlisting}
Thus, we can delegate to \texttt{chainApp} the task of appropriately nesting the calls to \texttt{lapp}; indeed, the value computed by \texttt{chainApp([R1,R2,...,Rn])} is the rule set given by \texttt{lapp((...lapp(} \texttt{R1,R2)...),Rn)}.

\subsection{Implementing replication}\label{app:subsec:bang}
The last operator we need is replication: $!f=\{(\langle t,u \rangle,\langle t,v\rangle)\mid t\in T_{\Sigma} \wedge (u,v)\in f\}$. Its implementation in Erlang is straightforward:
\begin{lstlisting}
bang(F) ->
  Vars=ruleListVars(F),
  X=fresh(Vars),
  bang_rec(F,X).

bang_rec(F,X) ->
  case F of
    [] -> [];
    [{R1,R2} | Tail] -> [{{p,{var,X},R1},{p,{var,X},R2}}]++bang_rec(Tail,X)
  end.
\end{lstlisting}
Notice how the fresh variable \texttt{X} plays the role of the generic term $t$ in the original definition. Indeed, being a new variable, \texttt{X} can be unified with every possible term.

\section{From $\lambda$-terms to combinators}\label{sec:lambda-comb}
In this section we will illustrate the details of the functions needed to transform plain $\lambda-terms$ to expressions involving only combinators.

\subsection{Auxiliary functions}

\subsubsection{Renaming automatic variables}\label{subsubsec:shift}

In order to do not interfere with variable names assigned by users, we reserve identifiers of the shape \texttt{\_Xi} (where \texttt{i} is a natural number $\geq 1$) for automatic variables (i.e., variables generated in a programmatic way). Thus, we can easily rename such variables by increasing the index \texttt{i}. This is precisely the purpose of the function \texttt{shift} which increases the index of the first argument by the amount specified by the second argument:

\begin{lstlisting}
shift(Var,I) ->
  Len=length(Var),
  if
    Len>=2 ->
      Prefix=string:substr(Var, 1, 2),
      if
        Prefix=="_X" ->
          {VarId,_}=string:to_integer(string:substr(Var, 3, Len)),
          if
            VarId==error -> Var;
            true -> lists:flatten(io_lib:format("_X~p", [VarId+I]))
          end;
        true -> Var
      end;
    true -> Var
  end.
\end{lstlisting}

\subsubsection{Function computing the substitution needed to rename (with \texttt{Fv}, \texttt{shift} \texttt{(Fv,1)}, \texttt{shift(Fv,2)}, \ldots) all variables occurring in a list of variables}

The function \texttt{compute\_renaming} uses the previously defined function \texttt{shift} (see Section~\ref{subsubsec:shift}), in order to transform a list of variables \verb|[X, Y, Z, ...]| into \verb|[Fv,| \verb|shift(Fv,1), shift(Fv,2), ...]|.

\begin{lstlisting}
compute_renaming(Vars,Fv) ->
    case Vars of
      [] -> [];
      [V|Tail] -> [{V,{var,Fv}} | compute_renaming(Tail,shift(Fv,1))]
    end.
\end{lstlisting}

\subsubsection{Function renaming (with \texttt{Fv}, \texttt{shift(Fv,1)}, \texttt{shift(Fv,2)}, \ldots) all variables occurring in the ruleset \texttt{R}}

Using \texttt{compute\_renaming} we can proceed further with the renaming of all the variables occurring in a ruleset by \texttt{Fv}, \texttt{shift(Fv,1)}, \texttt{shift(Fv,2)}, \ldots.

\begin{lstlisting}
polish_rules(R,Fv) ->
    case R of
      [] -> [];
      [{R1,R2}|Tail] -> Vars=ruleListVars([{R1,R2}]),
                        Subst=compute_renaming(Vars,Fv),
                        [{subListTerm(Subst,R1),subListTerm(Subst,R2)} | polish_rules(Tail,Fv)]
    end.
\end{lstlisting}

\subsubsection{Polishing a ruleset \texttt{R}}

The following function generates a new fresh variable and it uses the latter as a basis for renaming all variables of the ruleset \texttt{R}, removing also possible duplicate clauses.

\begin{lstlisting}
polish(R) ->
    Vars=ruleListVars(R),
    Fv=fresh(Vars),
    dd(polish_rules(R,Fv)).
\end{lstlisting}

\subsubsection{Deciding if two rulesets are equivalent}

So far, it is rather straightforward to check if two given rulesets \texttt{R1} and \texttt{R2} are equivalent:
\begin{enumerate}
\item choose a fresh variable w.r.t. both variables in \texttt{R1}, \texttt{R2};
\item polish both rulesets, removing possible duplicates;
\item test if the results are equal.
\end{enumerate}

\begin{lstlisting}
equiv(R1,R2) ->
  Vars1=ruleListVars(R1),
  Vars2=ruleListVars(R2),
  Fv=fresh(Vars1++Vars2),
  P1=dd(polish_rules(R1,Fv)),
  P2=dd(polish_rules(R2,Fv)),
  lists:sort(P1) =:= lists:sort(P2).
\end{lstlisting}

\subsubsection{Computing the list of the free variables of a $\lambda$-term}

Our last auxiliary function computes the list of the free variables occurring in a $\lambda$-term.
\begin{lstlisting}
varsLambda(Lambda) ->
  case Lambda of
  	{var, X} -> [X];
    {comb, _} -> [];
    {lapp, M, N} -> dd(varsLambda(M) ++ varsLambda(N));
    {abs, {var, X}, M} -> lists:delete(X,varsLambda(M));
    {bang, M} -> varsLambda(M)
  end.
\end{lstlisting}

\subsection{Abstraction operator}
The implementation of the \emph{Abstraction Operator}, extended to the full $\lambda^!$-calculus, allows the user to ``translate'' a $\lambda^!$-term into an expression built using only combinators. This is precisely the purpose of \texttt{abstract}:

\begin{lstlisting}[
    basicstyle=\tiny,
]
abstract(Lambda) ->
  case Lambda of
    %%%%%%%%%%%%%%%%%%%%%%
    % lambda*! X.C -> KC %
    %%%%%%%%%%%%%%%%%%%%%%
    {abs_b, {var, _}, {comb, C}} -> {lapp, {comb, "K"}, {comb, C}};
    %%%%%%%%%%%%%%%%%%%%%%%%%%%%%%%%
    % lambda*! X.Y -> D (if X=Y),  %
    % lambda*! X.Y -> KY otherwise %
    %%%%%%%%%%%%%%%%%%%%%%%%%%%%%%%%
    {abs_b, {var, X}, {var, Y}} ->
      if
        X == Y -> {comb, "D"};
        true -> {lapp, {comb, "K"}, {var, Y}}
      end;
    %%%%%%%%%%%%%%%%%%%%%%%%%%
    % lambda*! X.!X -> F(!I) %
    %%%%%%%%%%%%%%%%%%%%%%%%%%
    {abs_b, {var, X}, {bang, {var, X}}} -> {lapp, {comb, "F"}, {bang, {comb, "I"}}};
    %%%%%%%%%%%%%%%%%%%%%%%%%%%%%%%%%%%%%%%%%%%%%%%%%%%%%%%%%%%%%%%%%%%%%%%%%%
    % lambda*! X.MN -> C(lambda* X.M)N (if X in FV(M) and X not in FV(N))    %
    % lambda*! X.MN -> BM(lambda* X.N) (if X not in FV(M) and X in FV(N))    %
    % lambda*! X.MN -> W(C(BB(lambda*! X.M))(lambda*! X.N)) (if X in FV(M)   %
    %                                                        and X in FV(N)) %
    % lambda*! X.MN -> K(MN) otherwise                                       %
    %%%%%%%%%%%%%%%%%%%%%%%%%%%%%%%%%%%%%%%%%%%%%%%%%%%%%%%%%%%%%%%%%%%%%%%%%%
    {abs_b, {var, X}, {lapp, M, N}} ->
      CheckM=lists:member(X,varsLambda(M)),
      CheckN=lists:member(X,varsLambda(N)),
      if
        CheckM and not(CheckN) ->
          {lapp, {lapp, {comb, "C"}, abstract({abs_b, {var, X}, M})}, N};
        not(CheckM) and CheckN ->
          {lapp, {lapp, {comb, "B"}, M}, abstract({abs_b, {var, X}, N})};
        CheckM and CheckN ->
          {lapp, {comb, "W"},
                 {lapp, {lapp, {comb, "C"},
                               {lapp, {lapp, {comb, "B"}, {comb, "B"} },
                                      abstract({abs_b, {var, X}, M})
                               }
                        },
                        abstract({abs_b, {var, X}, N})
                 }
          };
        true -> {lapp, {comb, "K"}, {lapp, M, N}}
      end;
    %%%%%%%%%%%%%%%%%%%%%%%%%%%%%%%%%%%%%%%%%%%%%%%%%%%%%%
    % lambda*! X.!M -> B(F(!lambda*!X.M))delta (if M<>X) %
    %%%%%%%%%%%%%%%%%%%%%%%%%%%%%%%%%%%%%%%%%%%%%%%%%%%%%%
    {abs_b, {var, X}, {bang, M}} when M /= {var, X} ->
      {lapp, {lapp, {comb, "B"}, {lapp, {comb, "F"}, {bang, abstract({abs_b, {var, X}, M})}}}, {comb, "d"}};
    %%%%%%%%%%%%%%%%%%%%%
    % lambda* X.C -> KC %
    %%%%%%%%%%%%%%%%%%%%%
	{abs, {var, _}, {comb, C}} -> {lapp, {comb, "K"}, {comb, C}};
    %%%%%%%%%%%%%%%%%%%%%%%%%%%%%%%
    % lambda* X.Y -> I (if X=Y),  %
    % lambda* X.Y -> KY otherwise %
    %%%%%%%%%%%%%%%%%%%%%%%%%%%%%%%
    {abs, {var, X}, {var, Y}} ->
      if
        X == Y -> {comb, "I"};
        true -> {lapp, {comb, "K"}, {var, Y}}
      end;
    %%%%%%%%%%%%%%%%%%%%%%%%%%%%%%%%%%%%%%%%%%%%%%%%%%%%%%%%%%%%%%%
    % lambda* X.MN -> C(lambda* X.M)N (if X in FV(M))             %
    % lambda* X.MN -> BM(lambda* X.N) (if X in FV(N))             %
    % lambda* X.MN -> K(MN) otherwise                             %
    %%%%%%%%%%%%%%%%%%%%%%%%%%%%%%%%%%%%%%%%%%%%%%%%%%%%%%%%%%%%%%%
    {abs, {var, X}, {lapp, M, N}} ->
      CheckM=lists:member(X,varsLambda(M)),
      if
        CheckM -> {lapp, {lapp, {comb, "C"}, abstract({abs, {var, X}, M})}, N};
        true ->
          CheckN=lists:member(X,varsLambda(N)),
          if
            CheckN -> {lapp, {lapp, {comb, "B"}, M}, abstract({abs, {var, X}, N})};
            true -> {lapp, {comb, "K"}, {lapp, M, N}}
          end
      end;
    %%%%%%%%%%%%%%%%%%%%%%%%%%%%%%%%%%%%%%%%
    % Structural abstraction rules (begin) %
    %%%%%%%%%%%%%%%%%%%%%%%%%%%%%%%%%%%%%%%%
    {abs, {var, X}, {abs, {var, Y}, M}} ->
      N=abstract({abs, {var, Y}, M}),
      abstract({abs, {var, X}, N});
    {abs, {var, X}, {abs_b, {var, Y}, M}} ->
      N=abstract({abs_b, {var, Y}, M}),
      abstract({abs, {var, X}, N});
    {abs_b, {var, X}, {abs, {var, Y}, M}} ->
      N=abstract({abs, {var, Y}, M}),
      abstract({abs_b, {var, X}, N});
    {abs_b, {var, X}, {abs_b, {var, Y}, M}} ->
      N=abstract({abs_b, {var, Y}, M}),
      abstract({abs_b, {var, X}, N});
    %%%%%%%%%%%%%%%%%%%%%%%%%%%%%%%%%%%%%%
    % Structural abstraction rules (end) %
    %%%%%%%%%%%%%%%%%%%%%%%%%%%%%%%%%%%%%%
    %%%%%%%%%%%%%%%%%%%%%%%%%%%%%%%%
    % If nothing else applies, ... %
    %%%%%%%%%%%%%%%%%%%%%%%%%%%%%%%%
    M -> M
  end.
\end{lstlisting}

\section{Combinators as partial involutions}

\subsection{Standard Combinators}
The purpose of the following function is to return a list whose members are the internal representation as partial involutions of the standard combinators of Combinatory Logic.

\begin{lstlisting}
export_combinators() ->
  % Combinators as sets of rewriting rules.
  String_I="lX<->rX",
  String_B="rrrX<->lrX, llX<->rlrX, rllX<->rrlX",
  String_K="lX<->rrX",
  String_C="llX<->rrlX, lrlX<->rlX, lrrX<->rrrX",
  String_D="l<e,X><->rX",
  String_F="l<X,rY><->rr<X,Y>, l<X,lY><->rl<X,Y>",
  String_W="rrX<->lrrX, ll<X,Y><->rl<lX,Y>, lrl<X,Y><->rl<rX,Y>",
  String_Delta="l<<X,Y>,Z><->r<X,<Y,Z>>",
  String_B1="rrrX<->lrX, llX<->rlrX, rll<X,Y><->rrl<X,Y>",

  % Lexing and parsing.
  {ok,Tokens_I,1}=abramsky_lexer:string(String_I),
  {ok,I}=abramsky_parser:parse(Tokens_I),
  {ok,Tokens_B,1}=abramsky_lexer:string(String_B),
  {ok,B}=abramsky_parser:parse(Tokens_B),
  {ok,Tokens_K,1}=abramsky_lexer:string(String_K),
  {ok,K}=abramsky_parser:parse(Tokens_K),
  {ok,Tokens_C,1}=abramsky_lexer:string(String_C),
  {ok,C}=abramsky_parser:parse(Tokens_C),
  {ok,Tokens_D,1}=abramsky_lexer:string(String_D),
  {ok,D}=abramsky_parser:parse(Tokens_D),
  {ok,Tokens_F,1}=abramsky_lexer:string(String_F),
  {ok,F}=abramsky_parser:parse(Tokens_F),
  {ok,Tokens_W,1}=abramsky_lexer:string(String_W),
  {ok,W}=abramsky_parser:parse(Tokens_W),
  {ok,Tokens_Delta,1}=abramsky_lexer:string(String_Delta),
  {ok,Delta}=abramsky_parser:parse(Tokens_Delta),
  {ok,Tokens_B1,1}=abramsky_lexer:string(String_B1),
  {ok,B1}=abramsky_parser:parse(Tokens_B1),

  % Returning the list
  [{"I",I},{"B",B},{"K",K},{"C",C},{"D",D},{"F",F},{"W",W},{"d",Delta},{"B1",B1}].
\end{lstlisting}

\subsection{Decoding combinators expressions into partial involutions}
The result yielded by \texttt{export\_combinators} can be fed to \texttt{decode} in order to
transform a combinatory logic term into linear and bang applications (yielding a partial involution as the final result).

\begin{lstlisting}
decode(TComb,Combinators) ->
  case TComb of
    {comb, Comb} -> Check=lists:keyfind(Comb,1,Combinators),
                    if
                      Check==false -> [];
                      true -> element(2,Check)
                    end;
    {lapp, M, N} -> lapp(decode(M,Combinators),decode(N,Combinators));
    {bang, M} -> bang(decode(M,Combinators))
  end.
\end{lstlisting}

\section{Testing and pretty printing}

In this section we describe some auxiliary functions whose purpose it to ease the automation of tests (i.e., the equivalence checks between partial involutions), and to pretty print the results to the screen.

\subsection{Pretty printing terms}

\begin{lstlisting}
pretty_print_term(T) ->
  case T of
    e -> io:format("e");
    {var, V} -> io:format(V);
    {l,T1} -> io:format("l("), pretty_print_term(T1), io:format(")");
    {r,T1} -> io:format("r("), pretty_print_term(T1), io:format(")");
    {p,T1,T2} -> io:format("<"), pretty_print_term(T1),io:format(", "), pretty_print_term(T2), io:format(">")
  end.
\end{lstlisting}

\subsection{Pretty printing lists of rewriting rules}

\begin{lstlisting}
pretty_print_rules(L) ->
  case L of
    [] -> io:format("--- end ---~n");
    [{R1,R2} | T] -> pretty_print_term(R1), io:format(" -> "), pretty_print_term(R2), io:format("~n"), pretty_print_rules(T)
  end.
\end{lstlisting}

\subsection{Pretty printing lambda terms}

\begin{lstlisting}
pretty_print_lambda(T) ->
  case T of
    {var, X} -> io:format(X);
    {comb, C} -> io:format(C);
    {lapp, M, N} -> io:format("("),pretty_print_lambda(M),pretty_print_lambda(N),io:format(")");
    {abs, {var, X}, M} -> io:format("l* "),io:format(X),io:format("."),pretty_print_lambda(M);
    {abs_b, {var, X}, M} -> io:format("l*! "),io:format(X),io:format("."),pretty_print_lambda(M);
    {bang, M} -> io:format("!("),pretty_print_lambda(M),io:format(")")
  end.
\end{lstlisting}

\subsection{Checking equivalence of partial involutions}

The function \texttt{test} in the following returns \texttt{true} if and only if \texttt{Lambda\_string1} and \texttt{Lambda\_string2} are equivalent $\lambda^!$-expressions.

\begin{lstlisting}
test(Lambda_string1,Lambda_string2,Combinators) ->
  {ok,Tokens_LambdaT1,1}=lambda_lexer:string(Lambda_string1),
  {ok,LambdaT1}=lambda_parser:parse(Tokens_LambdaT1),
  Comb_LambdaT1=abstract(LambdaT1),
  io:format("--- "),pretty_print_lambda(Comb_LambdaT1),io:format(" ---~n"),
  Test1=decode(Comb_LambdaT1,Combinators),
  pretty_print_rules(Test1),
  {ok,Tokens_LambdaT2,1}=lambda_lexer:string(Lambda_string2),
  {ok,LambdaT2}=lambda_parser:parse(Tokens_LambdaT2),
  Comb_LambdaT2=abstract(LambdaT2),
  io:format("--- "),pretty_print_lambda(Comb_LambdaT2),io:format(" ---~n"),
  Test2=decode(Comb_LambdaT2,Combinators),
  pretty_print_rules(Test2),
  equiv(Test1,Test2).
\end{lstlisting}

\subsection{Converting and pretty printing $\lambda^!$-terms}
The following function prints to the screen the combinators expression and  the rewriting rules corresponding to the $\lambda^!$-expression \texttt{LambdaString}.

\begin{lstlisting}
show(LambdaString) ->
  Combinators=export_combinators(),
  {ok,Tokens,1}=lambda_lexer:string(LambdaString),
  {ok,T}=lambda_parser:parse(Tokens),
  Comb_T=abstract(T),
  pretty_print_lambda(Comb_T), io:format("~n"),
  pretty_print_rules(polish(decode(Comb_T,Combinators))).
\end{lstlisting}

\section{Running example}
In order to give an idea of how to use all the machinery so far introduced, we will consider an example session at the Erlang console.

Let us see how to prove that equation $\lambda^*xyz.C(C(BBx)y)z = \lambda^*xyz.Cx(yz)$ from Theorem~\ref{th:eqforSALA} holds:

\begin{lstlisting}
1> Combinators=abramsky:export_combinators().
[{"I",
  [{{l,{var,"X"}},{r,{var,"X"}}},
   {{r,{var,"X"}},{l,{var,"X"}}}]},
 {"B",
  [{{r,{r,{r,{var,"X"}}}},{l,{r,{var,"X"}}}},
   {{l,{r,{var,"X"}}},{r,{r,{r,{var,"X"}}}}},
   {{l,{l,{var,"X"}}},{r,{l,{r,{var,"X"}}}}},
   {{r,{l,{r,{var,"X"}}}},{l,{l,{var,"X"}}}},
   {{r,{l,{l,{var,"X"}}}},{r,{r,{l,{var,"X"}}}}},
   {{r,{r,{l,{var,"X"}}}},{r,{l,{l,{var,"X"}}}}}]},
 {"K",
  [{{l,{var,"X"}},{r,{r,{var,"X"}}}},
   {{r,{r,{var,"X"}}},{l,{var,"X"}}}]},
 {"C",
  [{{l,{l,{var,"X"}}},{r,{r,{l,{var,"X"}}}}},
   {{r,{r,{l,{var,"X"}}}},{l,{l,{var,"X"}}}},
   {{l,{r,{l,{var,"X"}}}},{r,{l,{var,"X"}}}},
   {{r,{l,{var,"X"}}},{l,{r,{l,{var,"X"}}}}},
   {{l,{r,{r,{var,"X"}}}},{r,{r,{r,{var,"X"}}}}},
   {{r,{r,{r,{var,"X"}}}},{l,{r,{r,{var,"X"}}}}}]},
 {"D",
  [{{l,{p,e,{var,"X"}}},{r,{var,"X"}}},
   {{r,{var,"X"}},{l,{p,e,{var,"X"}}}}]},
 {"F",
  [{{l,{p,{var,"X"},{r,{var,"Y"}}}},
    {r,{r,{p,{var,"X"},{var,"Y"}}}}},
   {{r,{r,{p,{var,"X"},{var,"Y"}}}},
    {l,{p,{var,"X"},{r,{var,"Y"}}}}},
   {{l,{p,{var,"X"},{l,{var,"Y"}}}},
    {r,{l,{p,{var,"X"},{var,"Y"}}}}},
   {{r,{l,{p,{var,"X"},{var,"Y"}}}},
    {l,{p,{var,"X"},{l,{var,"Y"}}}}}]},
 {"W",
  [{{r,{r,{var,"X"}}},{l,{r,{r,{var,"X"}}}}},
   {{l,{r,{r,{var,"X"}}}},{r,{r,{var,"X"}}}},
   {{l,{l,{p,{var,"X"},{var,"Y"}}}},
    {r,{l,{p,{l,{var,"X"}},{var,"Y"}}}}},
   {{r,{l,{p,{l,{var,"X"}},{var,"Y"}}}},
    {l,{l,{p,{var,"X"},{var,"Y"}}}}},
   {{l,{r,{l,{p,{var,"X"},{var,"Y"}}}}},
    {r,{l,{p,{r,{var,"X"}},{var,"Y"}}}}},
   {{r,{l,{p,{r,{var,"X"}},{var,"Y"}}}},
    {l,{r,{l,{p,{var,"X"},{var,"Y"}}}}}}]},
 {"d",
  [{{l,{p,{p,{var,"X"},{var,"Y"}},{var,"Z"}}},
    {r,{p,{var,"X"},{p,{var,"Y"},{var,"Z"}}}}},
   {{r,{p,{var,"X"},{p,{var,"Y"},{var,"Z"}}}},
    {l,{p,{p,{var,"X"},{var,"Y"}},{var,"Z"}}}}]},
 {"B1",
  [{{r,{r,{r,{var,"X"}}}},{l,{r,{var,"X"}}}},
   {{l,{r,{var,"X"}}},{r,{r,{r,{var,"X"}}}}},
   {{l,{l,{var,"X"}}},{r,{l,{r,{var,"X"}}}}},
   {{r,{l,{r,{var,"X"}}}},{l,{l,{var,"X"}}}},
   {{r,{l,{l,{p,{var,"X"},{var,"Y"}}}}},
    {r,{r,{l,{p,{var,"X"},{var,[...]}}}}}},
   {{r,{r,{l,{p,{var,"X"},{var,[...]}}}}},
    {r,{l,{l,{p,{var,[...]},{var,...}}}}}}]}]
2> LambdaT1_string="l* X.l* Y. l* Z.C@(C@(B@B@X)@Y)@Z". (*@\label{code:example:LambdaT1_string}@*)
"l* X.l* Y. l* Z.C@(C@(B@B@X)@Y)@Z"
3> LambdaT2_string="l* X.l* Y. l* Z.C@X@(Y@Z)".
"l* X.l* Y. l* Z.C@X@(Y@Z)" (*@\label{code:example:LambdaT2_string}@*)
4> io:format("Test: lambda*xyz.C(C(BBx)y)z=lambda*xyz.Cx(yz): ~p~n",[abramsky:test(LambdaT1_string,LambdaT2_string,Combinators)]). (*@\label{code:test_call}@*)
--- ((C((BC)((B(BB))((B(BC))((C((BB)((BC)((B(BB))I))))I)))))I) ---
l(l(_X8)) -> r(r(r(l(_X8))))
l(r(r(_X8))) -> r(r(r(r(_X8))))
l(r(l(_X7))) -> r(l(r(_X7)))
r(l(r(_X5))) -> l(r(l(_X5)))
r(r(r(l(_X4)))) -> l(l(_X4))
r(r(r(r(_X5)))) -> l(r(r(_X5)))
r(l(l(_X9))) -> r(r(l(_X9)))
r(r(l(_X7))) -> r(l(l(_X7)))
--- end ---
--- ((C((BB)((BB)((BC)I))))((C((BB)I))I)) ---
l(l(_X4)) -> r(r(r(l(_X4))))
l(r(r(_X4))) -> r(r(r(r(_X4))))
r(r(r(l(_X3)))) -> l(l(_X3))
r(r(r(r(_X3)))) -> l(r(r(_X3)))
l(r(l(_X7))) -> r(l(r(_X7)))
r(l(r(_X3))) -> l(r(l(_X3)))
r(l(l(_X5))) -> r(r(l(_X5)))
r(r(l(_X7))) -> r(l(l(_X7)))
--- end ---
Test: lambda*xyz.C(C(BBx)y)z=lambda*xyz.Cx(yz): true (*@\label{code:test_result}@*)
ok
5>
\end{lstlisting}
Hence, the equation holds; indeed, the value returned by the call to \texttt{abramsky:test} \texttt{(LambdaT1\_string,LambdaT2\_string,Combinators)} at line~\ref{code:test_call} is \texttt{true} (line~\ref{code:test_result}).

Notice again (see Section~\ref{web-app}) that the rewriting rules of partial involutions (\texttt{A <-> B}) are, for programming convenience reasons, internally represented by two expressions, namely, \texttt{A -> B} and \texttt{B -> A}. Moreover, in lines~\ref{code:example:LambdaT1_string} --\ref{code:example:LambdaT2_string} you can see our concrete syntax for $\lambda$-terms: in particular, $\lambda^*$ corresponds to \texttt{l*}, and application between terms is rendered by \verb|@|.

\end{appendices}

\end{document}